\documentclass{llncs}

\usepackage{thmtools,thm-restate}

\usepackage{amsmath}
\usepackage{amsthm}
\usepackage{amssymb}
\usepackage{tikz}
\usetikzlibrary{positioning,automata,fit,shapes}
\usepackage{dsfont}
\usepackage{enumitem}

\usepackage[section]{placeins}
\newcommand{\tudparagraph}[2]{%
\vspace*{#1}

\noindent
{\bf #2}
}

\newcommand{\cE}{\mathcal{E}}

\newcommand{\cM}{\mathcal{M}}
\newcommand{\cN}{\mathcal{N}}

\renewcommand{\Pr}{\mathrm{Pr}}

\newcommand{\sinit}{s_{\mathit{\scriptscriptstyle init}}}

\newcommand{\Act}{\mathit{Act}}

\newcommand{\act}{\alpha}

\newcommand{\Size}{\mathit{size}}

\newcommand{\fpath}{\pi}

\newcommand{\infpath}{\varsigma}

\newcommand{\last}{\mathit{last}}

\newcommand{\sched}{\mathfrak{S}}

\newcommand{\MEC}{\mathit{MEC}}

\newcommand{\wgt}{\mathit{wgt}}

\newcommand{\ActEC}{\mathfrak{A}}

\newcommand{\goal}{\mathit{goal}}
\newcommand{\fail}{\mathit{fail}}

\newcommand{\Rational}{\mathbb{Q}}

\newcommand{\CiteAppendix}[1]{}

\theoremstyle{definition}
\newtheorem{mydef}{Definition}

\theoremstyle{plain}
\newtheorem{myprop}[mydef]{Proposition}

\theoremstyle{definition}
\newtheorem{myex}[mydef]{Example}

\theoremstyle{plain}
\newtheorem{mylem}[mydef]{Lemma}

\theoremstyle{plain}
\newtheorem{mycor}[mydef]{Corollary}

\theoremstyle{plain}
\newtheorem{mythm}[mydef]{Theorem}

\theoremstyle{definition}

\begin{document}

\title{Partial and Conditional Expectations in {M}arkov Decision Processes with Integer Weights}%
\author{Jakob Piribauer \and Christel Baier}%
\institute{Technische Universit{\"a}t Dresden\thanks{%
The authors are supported by the DFG through the Research Training Group QuantLA (GRK 1763),
the DFG-project BA-1679/11-1,
    the Collaborative Research Center HAEC (SFB 912),
    and the cluster of excellence CeTI.}
    \\
    \email{ \{jakob.piribauer,christel.baier\}@tu-dresden.de}
    }

\maketitle
\thispagestyle{plain}


\begin{abstract}
The paper addresses two variants of the stochastic shortest path problem (``optimize the accumulated weight until reaching a goal state'') in Markov decision processes (MDPs) with integer weights.
The first variant optimizes partial expected accumulated weights, where paths not leading to a goal state are assigned weight $0$, while the second variant considers conditional expected accumulated weights, where the probability mass is redistributed to paths reaching the goal. Both variants constitute useful approaches to the analysis of systems without guarantees on the occurrence of an event of interest (reaching a goal state), but have only been studied in structures with non-negative weights. Our main results are as follows.
There are polynomial-time algorithms to check the
finiteness of the supremum of the partial or conditional expectations
in MDPs with arbitrary integer weights.
If finite, then optimal weight-based deterministic schedulers exist.
In contrast to the setting of non-negative weights,
optimal schedulers can need infinite memory and
their value can be irrational.
However, the optimal value can be approximated up to an absolute
error of $\epsilon$ in time exponential in the size of the MDP
and polynomial in  $\log(1/\epsilon)$. \end{abstract}

\section{Introduction}


Stochastic shortest path (SSP) problems generalize the  shortest path problem on  graphs with weighted edges. 
The SSP problem is formalized using finite state Markov decision processes (MDPs),
which are a prominent model combining probabilistic and nondeterministic choices.
In each state of an MDP, one  is allowed to choose nondeterministically from a set of actions, each of them is augmented with probability distributions over the successor states and a weight (cost or reward). The SSP problem asks for a policy to choose actions (here called a scheduler) maximizing or minimizing the expected accumulated weight until reaching a goal state. In the classical setting, one seeks an optimal \emph{proper} scheduler where proper means
that a goal state is reached almost surely. Polynomial-time solutions exist
exploiting the fact that
optimal memoryless deterministic schedulers exist (provided the optimal value is finite)
and can be computed using linear programming techniques, possibly in combination
with model transformations
(see \cite{bertsekas1991,deAlfaro1999,lics2018}).
The restriction to proper schedulers, however, 
is often  too restrictive.
First, there are models that have no proper scheduler.
Second, even if proper schedulers exist,
the expectation of the accumulated weight of schedulers missing the goal with a positive probability should be taken into account as well.
Important such applications 
include the semantics of probabilistic programs (see e.g. \cite{gretz2014,katoen2015,barthe2016,chatterjee2016,olmedo2018}) where no guarantee for almost sure termination can be given and the analysis of program properties at termination time gives rise to stochastic shortest (longest) path problems in which the goal (halting configuration) is not reached almost surely.
Other examples are the fault-tolerance analysis (e.g., expected costs of repair mechanisms) in selected error scenarios that can appear with some positive, but small probability or the trade-off analysis with conjunctions of utility and cost constraints that are achievable with positive probability, but not almost surely (see e.g. \cite{baier2014}).

This motivates the switch to variants of classical SSP problems where the restriction to proper schedulers is relaxed. One option (e.g., considered in \cite{chen2013}) is to seek a scheduler optimizing the expectation of the random variable that
 assigns weight $0$ to all paths not reaching the goal and the accumulated weight of the
 shortest prefix reaching the goal  to all other paths.
We refer to this expectation as \emph{partial expectation}. Second, we consider the \emph{conditional expectation} of the accumulated weight until reaching the goal under the condition that the goal is reached.
In general, partial expectations describe situations in which some reward (positive and negative) is accumulated but only retrieved if a certain goal is met. 
In particular, partial expectations can be an appropriate replacement for the classical expected weight before reaching the goal if we want to include schedulers which miss the goal with some (possibly very small) probability. 
In contrast to conditional expectations, the resulting scheduler still has an incentive to reach the goal with a high probability, while schedulers maximizing the conditional expectation might reach the goal with a very small  positive probability.

Previous work on partial or conditional expected accumulated weights was restricted to the case of non-negative weights.
More precisely,
partial expectations have been studied in the setting of stochastic multiplayer games with non-negative weights \cite{chen2013}.
Conditional expectations in MDPs with non-negative weights 
have been addressed in \cite{tacas2017}. 
In both cases, optimal values are achieved by weight-based deterministic schedulers that depend on the current state and the weight that has been accumulated so far,  while memoryless schedulers are not sufficient.
Both \cite{chen2013} and \cite{tacas2017} prove the existence of a \emph{saturation point} for the accumulated weight from which on optimal schedulers behave memoryless and maximize the probability to reach a goal state.
This yields exponential-time algorithms for computing optimal schedulers using an iterative linear programming approach.
Moreover, \cite{tacas2017} proves that the threshold problem for conditional expectations (``does there exist a scheduler $\sched$ such that the conditional expectation under $\sched$ exceeds a given threshold?'') is PSPACE-hard even for acyclic MDPs. 

The purpose of the paper is to study partial and conditional expected accumulated weights for MDPs with integer weights.
The switch from non-negative to integer weights indeed causes several additional difficulties. We start with the following observation.
While optimal partial or conditional expectations in non-negative MDPs are rational, they can be irrational in the general setting:

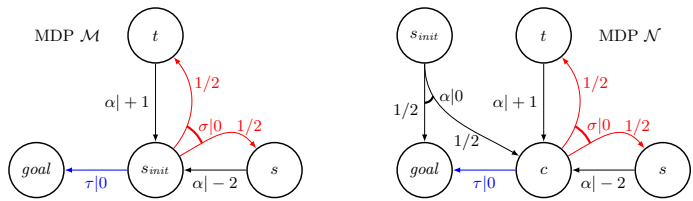
\begin{figure}
\begin{center}
\resizebox{.75\textwidth}{!}{%
\begin{tikzpicture}[scale=1,auto,node distance=8mm,>=latex,font=\small]
    \tikzstyle{round}=[thick,draw=black,circle]

    \node[round, draw=black,minimum size=30pt] (c) {$s_{\mathit{init}}$};
    \node[round, above=15mm of c,minimum size=30pt] (s3) {$t$};
        
    \node[round,right=12mm of c,minimum size=30pt] (s2) {$s$};
    \node[round,left=12mm of c,minimum size=30pt] (goal) {$\mathit{goal}$};
    
    \node[left=4mm of s3] (M) {MDP $\mathcal{M}$};

 \draw[color=red ,->] (c) [out=30,in=130,loop,min distance=7mm] edge  node [ near start, anchor=center] (n5) {} node [pos=0.9,above] {$1/2$} (s2) ;
    \draw[color=red ,->] (c)  [out=50,in=-50,loop,min distance=7mm] edge node [near start, anchor=center] (n6) {} node [pos=0.7,right] {$1/2$} (s3) ;
    \draw[color=red , very thick] (n5.center) edge [bend right=15] node [pos=0.95,right=2pt] {$\sigma | 0$} (n6.center);

    \draw[color=blue ,->] (c)  edge node {$\tau | 0$} (goal) ;
    \draw[color=black ,->] (s3)  edge node [left] {$\alpha | +1$} (c) ;
    \draw[color=black ,->] (s2)  edge node {$\alpha | -2$} (c) ;

\end{tikzpicture}
\hspace{40pt}
\begin{tikzpicture}[scale=1,auto,node distance=8mm,>=latex,font=\small]
    \tikzstyle{round}=[thick,draw=black,circle]

    \node[round, draw=black,minimum size=30pt] (c) {$c$};
    \node[round, above=15mm of c,minimum size=30pt] (s3) {$t$};
        
    \node[round,right=12mm of c,minimum size=30pt] (s2) {$s$};
    \node[round,left=12mm of c,minimum size=30pt] (goal) {$\mathit{goal}$};
    
    \node[round, draw=black, above=15mm  of goal, minimum size=0pt,minimum size=30pt] (sinit) {$s_{\mathit{init}}$};
    
    \node[right=4mm of s3] (N) {MDP $\mathcal{N}$};

\draw[color=black ,->] (sinit) [out=-90,in=150,loop,min distance=7mm] edge  node [ near start, anchor=center] (m5) {} node [pos=0.6,below] {$1/2$} (c) ;
    \draw[color=black ,->] (sinit)   edge node [, anchor=center] (m6) {} node [pos=0.6,left] {$1/2$} (goal) ;
    \draw[color=black , very thick] (m5.center) edge [bend left=35] node [pos=0.1, right] {$\alpha | 0$} (m6.center);

 \draw[color=red ,->] (c) [out=30,in=130,loop,min distance=7mm] edge  node [ near start, anchor=center] (n5) {} node [pos=0.9,above] {$1/2$} (s2) ;
    \draw[color=red ,->] (c)  [out=50,in=-50,loop,min distance=7mm] edge node [near start, anchor=center] (n6) {} node [pos=0.7,right] {$1/2$} (s3) ;
    \draw[color=red , very thick] (n5.center) edge [bend right=15] node [pos=1,right=1mm] {$\sigma | 0$} (n6.center);

    \draw[color=blue ,->] (c)  edge node {$\tau | 0$} (goal) ;
    \draw[color=black ,->] (s3)  edge node [left] {$\alpha | +1$} (c) ;
    \draw[color=black ,->] (s2)  edge node {$\alpha | -2$} (c) ;

\end{tikzpicture}
}
\end{center}
\caption{Enabled actions are denoted by Greek letters and the weight associated to the action is stated after the bar. Probabilistic choices are marked by a bold arc and transition probabilities are denoted next to the arrows. }\label{fig:irrational1}
\end{figure}

\begin{myex}
Consider the MDP $\mathcal{M}$ depicted on the left in Figure \ref{fig:irrational1}. In the initial state $s_{\mathit{init}}$, two actions are enabled. Action $\tau$ leads to $\mathit{goal}$ with probability $1$ and weight $0$. Action $\sigma$ leads to the states $s$ and $t$ with probability $1/2$ from where we will return to $s_{\mathit{init}}$ with weight $-2$ or $+1$, respectively.  The scheduler choosing $\tau$ immediately leads to an expected weight of $0$ and is optimal among schedulers reaching the goal almost surely.
As long as we choose $\sigma$ in $s_{\mathit{init}}$, the accumulated weight follows an asymmetric random walk increasing by $1$ or decreasing by $2$ with probability $1/2$ before we return to $s_{\mathit{init}}$. 
It is well known that 
the probability to ever reach accumulated weight $+1$ in this random walk is $1/\Phi$ where $\Phi=\frac{1+\sqrt{5}}{2}$ is the golden ratio. Likewise, ever reaching accumulated weight $n$ has probability $1/\Phi^n$ for all $n\in\mathbb{N}$. Consider the scheduler $\mathfrak{S}_k$ choosing $\tau$ as soon as the accumulated weight reaches $k$ in $s_{\mathit{init}}$. Its partial expectation is $k/\Phi^k$ as the paths which never reach weight $k$ are assigned weight $0$. The maximum is reached at $k=2$. In Section \ref{sec:optimalWD}, we prove that there are optimal schedulers whose decisions only depend on the current state and the weight accumulated so far. With this result we can conclude that the maximal partial expectation is indeed $2/\Phi^2$, an irrational number. 

The conditional expectation of $\mathfrak{S}_k$ in $\mathcal{M}$ is $k$ as $\mathfrak{S}_k$ reaches the goal with accumulated weight $k$ if it reaches the goal. So, the conditional expectation is not bounded. If we add a new initial state making sure that the goal is reached with positive probability as in the MDP $\mathcal{N}$, we can obtain an irrational maximal conditional expectation as well: The  scheduler $\mathfrak{T}_k$ choosing $\tau$ in $c$ as soon as the weight reaches $k$ has conditional expectation $\frac{k/2\Phi^k}{1/2+1/2\Phi^k}$. The maximum is obtained for $k=3$; the maximal conditional expectation is $\frac{3/\Phi^3}{1+1/\Phi^3}=\frac{3}{3+\sqrt{5}}$.
\end{myex}

Moreover,
while the proposed algorithms of \cite{chen2013,tacas2017} crucially rely on the monotonicity of the accumulated weights along the prefixes of paths, the accumulated weights of prefixes of path can oscillate when there are positive and negative weights. As we will see later, this implies that the existence of saturation points is no longer ensured and optimal schedulers might require infinite memory (more precisely, a counter for the accumulated weight). These observations provide evidence 
why linear-programming techniques as used in the
case of non-negative MDPs \cite{chen2013,tacas2017} 
cannot be expected to be applicable  for the general setting.

\tudparagraph{1ex}{Contributions.}
We study the problem of maximizing the partial and conditional expected
accumulated weight in MDPs with integer weights. 
Our first result is that the finiteness of the supremum of partial and conditional expectations in MDPs with integer weights can be checked in polynomial time (Section \ref{sec:PE}). 
For both variants we show that there are optimal weight-based deterministic schedulers if the supremum is finite (Section \ref{sec:optimalWD}).
Although the suprema might be irrational and optimal schedulers might need infinite memory, 
the suprema can be $\epsilon$-approximated in time exponential in the size of the MDP and polynomial in $\log(1/\epsilon)$ (Section \ref{sec:approximation}).
By duality of maximal and minimal expectations, 
analogous results hold for the problem of 
minimizing the partial or conditional expected accumulated weight.
(Note that we can multiply all weights by ${-}1$ and then apply the 
results for maximal partial resp. conditional expectations.)

\tudparagraph{1ex}{Related work.}
Closest to our contribution 
is the above mentioned work on partial expected accumulated weights
in stochastic multiplayer games 
with non-negative weights in \cite{chen2013} and on computation 
schemes for maximal conditional
expected accumulated weights in non-negative MDPs \cite{tacas2017}.
Conditional expected termination time in probabilistic push-down automata 
has been studied in \cite{esparza2005}, which can be seen as 
analogous considerations for a class of infinite-state Markov chains with
non-negative weights.
The recent work on notions of
conditional value at risk in MDPs \cite{kretinsky2018}
also studies conditional expectations, but the considered random variables
are limit averages and a notion of (non-accumulated) weight-bounded
reachability.


\section{Preliminaries} \label{sec:preliminaries}


We give basic definitions and present our notation. More details can be found in textbooks, e.g. \cite{puterman1994}.

\tudparagraph{1ex}{Notations for Markov decision processes.}
A \emph{Markov decision process} (MDP) is a tuple $\mathcal{M} = (S,\Act,P,\sinit,\wgt)$
where $S$ is a finite set of states,
$\Act$ a finite set of actions,
$\sinit \in S$ the initial state,
$P : S \times \Act \times S \to [0,1] \cap \Rational$ is the
transition probability function and
$\wgt : S \times \Act \to \mathbb{Z}$ the weight function.
We require that
$\sum_{t\in S}P(s,\act,t) \in \{0,1\}$
for all $(s,\alpha)\in S\times \Act$.
We write $\Act(s)$ for the set of actions that are enabled in $s$,
i.e., $\act \in \Act(s)$ iff $\sum_{t\in S}P(s,\act,t) =1$. We assume that $\Act(s)$ is non-empty for all $s$ and 
that all states are reachable from $s_{\mathit{init}}$. 
We call a state absorbing if the only enabled action leads to the state itself with probability $1$ and weight $0$.
The paths of $\cM$ are finite or
infinite sequences $s_0 \, \act_0 \, s_1 \, \act_1 \, s_2 \, \act_2 \ldots$
where states and actions alternate such that
$P(s_i,\act_i,s_{i+1}) >0$ for all $i\geq0$.
If $\fpath =
    s_0 \, \act_0 \, s_1 \, \act_1 \,  \ldots \act_{k-1} \, s_k$
is finite, then
  $\wgt(\fpath)=
   \wgt(s_0,\act_0) + \ldots + \wgt(s_{k-1},\act_{k-1})$
  denotes the accumulated weight of $\pi$,
    $P(\fpath) =
   P(s_0,\act_0,s_1) 
   \cdot \ldots \cdot P(s_{k-1},\act_{k-1},s_k)$
   its probability, and
$\last(\fpath)=s_k$ its last state.
The \emph{size} of $\cM$, denoted $\Size(\cM)$,
is the sum of the number of states
plus the total sum of the logarithmic lengths of the non-zero
probability values
$P(s,\alpha,s')$ as fractions of co-prime integers and the weight values $\wgt(s,\alpha)$.


\tudparagraph{1ex}{Scheduler.}
A \emph{(history-dependent, randomized) scheduler} for $\cM$
is a function $\sched$ that assigns to each finite path $\fpath$ 
a probability distribution over $\Act(\last(\fpath))$.
$\sched$ is called memoryless if $\sched(\fpath)=\sched(\fpath')$ for
all finite paths $\fpath$, $\fpath'$ with $\last(\fpath)=\last(\fpath')$,
in which case $\sched$ can be viewed as a function
that assigns to each state $s$ a distribution over $\Act(s)$.
$\sched$ is called deterministic if $\sched(\fpath)$ is a Dirac distribution
for each path $\fpath$,
in which case $\sched$ can be viewed as a function that assigns an action
to each finite path $\fpath$. 
Scheduler $\sched$ is said to be \emph{weight-based} if
$\sched(\fpath)=\sched(\fpath')$ for all finite paths $\fpath$, $\fpath'$
with $\wgt(\fpath)=\wgt(\fpath')$ and $\last(\fpath)=\last(\fpath')$.
Thus, deterministic weight-based schedulers
can be viewed as functions that assign
actions to state-weight-pairs.
By $\mathit{HR}^\mathcal{M}$ we denote the class of all schedulers, by 
$\mathit{WR}^\mathcal{M}$ the class of weight-based schedulers, 
by $\mathit{WD}^\mathcal{M}$ the class of weight-based, deterministic schedulers, 
and by $\mathit{MD}^\mathcal{M}$ the class of memoryless deterministic schedulers.
Given a scheduler $\sched$,
$\infpath \, = \, s_0 \, \act_0 \, s_1 \, \act_1 \ldots$
is a $\sched$-path iff $\infpath$ is a path and
$\sched(s_0 \, \act_0 \, s_1 \, \act_1 \ldots \act_{k-1} \, s_k)(\act_k)>0$
for all $k \geq 0$.

\tudparagraph{1ex}{Probability measure.}
We write $\Pr^{\sched}_{\cM,s}$ or briefly $\Pr^{\sched}_{s}$
to denote the probability measure induced by $\sched$ and $s$.
For details, see \cite{puterman1994}.
We will use LTL-like formulas to denote measurable sets of paths and also
write $\Diamond (wgt \bowtie x)$ to describe the set of infinite paths having a prefix $\pi$ 
with $wgt(\pi)\bowtie x$ for $x\in\mathbb{Z}$ and $\bowtie \ \in \{<,\leq,=,\geq,>\}$. 
Given a measurable set $\psi$ of infinite paths, we define
$\Pr^{\min}_{\cM,s}(\psi) = \inf_{\sched} \Pr^{\sched}_{\cM,s}(\psi)$
and
$\Pr^{\max}_{\cM,s}(\psi) = \sup_{\sched} \Pr^{\sched}_{\cM,s}(\psi)$
where $\sched$ ranges over all schedulers for $\cM$.
Throughout the paper, we suppose that the given MDP has
a designated state $\mathit{goal}$. Then, $p_s^{\max}$ and $p_s^{\min}$ denote the maximal resp. minimal probability of reaching $\mathit{goal}$ from $s$.
That is,
$p_s^{\max} = \mathrm{sup}_{\mathfrak{S}} \Pr^\mathfrak{S}_s (\Diamond \mathit{goal})$ 
and $p_s^{\min} = \inf_{\mathfrak{S}} \Pr^\mathfrak{S}_s (\Diamond \mathit{goal})$. Let 
$Act^{\max} (s) = \{\alpha \in Act(s) | \sum_{t\in S} P(s,\alpha,t) \cdot p_t^{\max} = p_s^{\max}\}$, 
and $Act^{\min} (s) = \{\alpha \in Act(s) | \sum_{t\in S} $ $P(s,\alpha,t) \cdot p_t^{\min} = p_s^{\min}\}$.

\tudparagraph{1ex}{Mean payoff.} A well-known measure for the long-run behavior of a scheduler $\mathfrak{S}$ in an MDP $\mathcal{M}$ is the \emph{mean payoff}. 
Intuitively, the mean payoff is the amount of weight accumulated per step on average in the long run. 
Formally, we define the mean payoff as the following random variable  on infinite paths $\zeta=s_0 \alpha_0 s_1 \alpha_1 \dots$:
$\mathit{MP}(\zeta):= \liminf\limits_{k\to \infty} \frac{\sum_{i=0}^k wgt(s_i,\alpha_i)}{k+1}$. 
The mean payoff of the scheduler $\mathfrak{S}$ starting in $s_{\mathit{init}}$ is then defined 
as the expected value  $\mathbb{E}^\mathfrak{S}_{s_{\mathit{init}}} (\mathit{MP})$. 
The maximal mean payoff is the supremum over all schedulers which is equal to the maximum over all $MD$-schedulers: 
$\mathbb{E}^{\max}_{s_{\mathit{init}}} (\mathit{MP})=\max_{\mathfrak{S}\in \mathit{MD}} \mathbb{E}^\mathfrak{S}_{s_{\mathit{init}}} (\mathit{MP})$. 
In strongly connected MDPs, the maximal mean payoff does not depend on the initial state.

\tudparagraph{1ex}{End components, MEC-quotient.}
An \emph{end component} of $\cM$ is a strongly connected sub-MDP. End components
can be formalized as pairs $\cE = (E,\ActEC)$ where $E$ is a nonempty subset
of $S$ and $\ActEC$ a function that assigns to each state $s\in E$ a nonempty
subset of $\Act(s)$ such that the graph induced by $\cE$ is strongly connected.
$\cE$ is called \emph{maximal} if there is no end component
$\cE' = (E',\ActEC')$ with $\cE \not= \cE'$, $E \subseteq E'$
and $\ActEC(s) \subseteq \ActEC'(s)$ for all $s\in E$.
The \emph{MEC-quotient} of an MDP $\cM$ is the MDP $\MEC(\cM)$ arising
from $\cM$
by collapsing all states that belong to the same maximal end component $\mathcal{E}$ to a state $s_\mathcal{E}$. 
All actions enabled in some state in $\mathcal{E}$ not belonging to $\mathcal{E}$ are enabled in $s_\mathcal{E}$. 
Details and the formal construction can be found in 
\cite{CBGK2008}. 
We call an end component $\mathcal{E}$
\emph{positively weight-divergent} if there is a scheduler $\mathfrak{S}$ for $\mathcal{E}$ 
such that $\Pr^\mathfrak{S}_{\mathcal{E},s}(\Diamond (\wgt \geq n))=1$ for all $s\in \mathcal{E}$ and $n\in \mathbb{N}$.  
In \cite{lics2018}, it is shown that the existence of positively weight-divergent end components can be decided in polynomial time.


\section{Partial and Conditional Expectations in MDPs} \label{sec:PE}


We define  \emph{partial} and \emph{conditional expectations} in MDPs. We extend the definition of \cite{chen2013}   by introducing partial expectations with \emph{bias} which  are closely related to conditional expectations.
Afterwards, we sketch the computation of  maximal partial expectations in MDPs with non-negative weights and in Markov chains.

\tudparagraph{1ex}{Partial and conditional expectation.} In the sequel, let $\mathcal{M}$ be an MDP with a designated absorbing goal state $\mathit{goal}$. Furthermore, we collapse all states from which $\mathit{goal}$ is not reachable to one absorbing state $\fail$.
Let  $b\in\mathbb{R}$. We define the random variable $\oplus^b \mathit{goal}$ on infinite paths $\zeta$ by
\[ \oplus^b \mathit{goal} (\zeta)=\begin{cases}
\wgt(\zeta) + b& \text{if }\zeta\vDash\Diamond \mathit{goal}\text{,}\\
0 & \text{if }\zeta\not\vDash\Diamond \mathit{goal}\text{.}
\end{cases}
\]
We call the expectation of this random variable under a scheduler $\mathfrak{S}$ the \emph{partial expectation with bias} $b$ of $\mathfrak{S}$ and write 
$\mathit{PE}^{\mathfrak{S}}_{\mathcal{M},s_{\mathit{init}}}[b] := \mathbb{E}^\mathfrak{S}_{\mathcal{M},s_{\mathit{init}}} (\oplus^b \mathit{goal})$ as well as $\mathit{PE}^{\sup}_{\mathcal{M},s_{\mathit{init}}}[b]:=\sup_{\mathfrak{S}\in \mathit{HR}^\mathcal{M}}\mathit{PE}^{\mathfrak{S}}_{\mathcal{M},s_{\mathit{init}}}[b]$. If $b=0$, we sometimes drop the argument $b$; if $\mathcal{M}$ is clear from the context, we drop the subscript. In order to maximize the partial expectation, intuitively one  has to find the right balance between reaching $\goal$ with  high probability and accumulating a high positive amount of weight before reaching $\goal$. The bias can be used to shift this balance by additionally rewarding or penalizing a high probability to reach $\goal$.

The \emph{conditional expectation} of $\mathfrak{S}$ is defined as the expectation of $\oplus^0 \mathit{goal}$ under the condition that $\mathit{goal}$ is reached. It  is defined if $\Pr^\mathfrak{S}_{\mathcal{M},s_{\mathit{init}}} (\Diamond \mathit{goal})>0$. We write
$\mathit{CE}^\mathfrak{S}_{\mathcal{M},s_{\mathit{init}}}:= \mathbb{E}^\mathfrak{S}_{\mathcal{M},s_{\mathit{init}}} (\oplus^0 \mathit{goal} | \Diamond \mathit{goal})$ and $\mathit{CE}^{\sup}_{\mathcal{M},s_{\mathit{init}}}=\sup_{\mathfrak{S}} \mathit{CE}^\mathfrak{S}_{\mathcal{M},s_{\mathit{init}}}$ where the supremum is taken over all schedulers $\mathfrak{S}$ with $\Pr^\mathfrak{S}_{\mathcal{M},s_{\mathit{init}}}(\Diamond \mathit{goal})>0$. 
We can express the conditional expectation as $\mathit{CE}^\mathfrak{S}_{\mathcal{M},s_{\mathit{init}}} = {\mathit{PE}^\mathfrak{S}_{\mathcal{M},s_{\mathit{init}}}}/{\Pr^\mathfrak{S}_{\mathcal{M},s_{\mathit{init}}}(\Diamond \mathit{goal})}$.
The following proposition establishes a close connection between conditional expectations and partial expectations with bias.

\begin{myprop} \label{prop:threshold}
Let $\mathcal{M}$ be an MDP, $\mathfrak{S}$ a scheduler with $\Pr^\mathfrak{S}_{\sinit}(\Diamond \goal)>0$, $\theta \in \mathbb{Q}$, and $\bowtie\in\{<,\leq,\geq,>\}$. Then we have
$\mathit{PE}^{\mathfrak{S}}_{s_{\mathit{init}}}[-\theta] \bowtie 0 \text{ iff }\mathit{CE}^{\mathfrak{S}}_{s_{\mathit{init}}}\bowtie \theta$.
Further, if $\Pr^{\min}_{s_{\mathit{init}}}(\Diamond \mathit{goal})>0$, then
$\mathit{PE}^{\sup}_{s_{\mathit{init}}} [-\theta]\bowtie 0 \text{ iff }\mathit{CE}^{\sup}_{s_{\mathit{init}}}\bowtie \theta$.
\end{myprop}

\begin{proof}
The first claim follows from $\mathit{PE}^{\mathfrak{S}}_{s_{\mathit{init}}}[-\theta]=\mathit{PE}^\mathfrak{S}_{s_{\mathit{init}}}[0]-\Pr^{\mathfrak{S}}_{s_{\mathit{init}}}(\Diamond \mathit{goal}) \cdot \theta$. The second claim follows by quantification over all schedulers. 
\end{proof}

In \cite{tacas2017}, it is shown that deciding whether $\mathit{CE}^{\sup}_{s_{\mathit{init}}}\bowtie \theta$ for $\bowtie\in\{<,\leq,\geq,>\}$ and $\theta\in\mathbb{Q}$ is PSPACE-hard even for acyclic MDPs. We conclude:
\begin{mycor}
Given an MDP $\mathcal{M}$, $\bowtie\in\{<,\leq,\geq,>\}$, and $\theta\in\mathbb{Q}$, deciding whether $\mathit{PE}^{\sup}_{\mathcal{M},s_{\mathit{init}}} \bowtie \theta$   is PSPACE-hard.
\end{mycor}


\tudparagraph{1ex.}{Finiteness.}
We present criteria for the finiteness of $\mathit{PE}^{\sup}_{s_{\mathit{init}}}[b]$ and $\mathit{CE}^{\sup}_{s_{\mathit{init}}}$.  
Detailed proofs can be found in Appendix \ref{app:finiteness}. 
By slightly modifying the construction from \cite{lics2018} which removes end components only containing $0$-weight cycles, we obtain the following result.
\begin{restatable}{myprop}{removingzeroECs} \label{prop:removingzeroECs}
Let $\cM$ be an MDP which does not contain positively weight-divergent end components and let $b\in\mathbb{Q}$. Then there is a polynomial time transformation to an MDP $\cN$ containing all states from $\cM$ and possibly an additional absorbing state $\mathit{fail}$ such that 
\begin{itemize}[noitemsep,topsep=0pt]
\item all end components of $\cN$ have negative maximal expected mean payoff,
\item for any scheduler $\sched$ for $\cM$ there is a scheduler $\sched^\prime$ for $\cN$ with $\Pr^\sched_{\cM,s}(\Diamond \mathit{goal})=\Pr^{\sched^\prime}_{\cN,s}(\Diamond \mathit{goal})$ and $\mathit{PE}^{\sched}_{\cM,s}[b]=\mathit{PE}^{\sched^\prime}_{\cN,s}[b]$ for any state $s$ in $\cM$, and vice versa.
\end{itemize}
\end{restatable}
Hence, we can restrict ourselves to MDPs in which all end components have negative maximal expected mean payoff if there are no positively weight divergent end components.
The following result is now analogous to the result in \cite{lics2018} for the classical SSP problem.

\begin{restatable}{myprop}{finitenessPE} Let $\mathcal{M}$ be an MDP and $b\in\mathbb{R}$ arbitrary.
The optimal partial expectation $\mathit{PE}^{\sup}_{s_{\mathit{init}}}[b]$ is finite if and only if there are no positively weight-divergent end components in $\mathcal{M}$.
\end{restatable}

To obtain an analogous result for conditional expectations, we observe that the finiteness of the maximal partial expectation is necessary for the finiteness of the maximal conditional expectation. However, this is not sufficient.
 In \cite{tacas2017}, a \emph{critical scheduler} is defined as a scheduler $\mathfrak{S}$ for which there is a path containing a positive cycle and for which $\Pr^{\mathfrak{S}}_{s_{\mathit{init}}}(\Diamond \mathit{goal})=0$. Given a critical scheduler, it is easy to construct a sequence of schedulers with unbounded conditional expectation (see Appendix \ref{app:finiteness} and \cite{tacas2017}).
On the other hand, if $\Pr^{\min}_{\mathcal{M},s_{\mathit{init}}}(\Diamond \mathit{goal})>0$, then $\mathit{CE}^{\sup}_{s_{\mathit{init}}}$ is finite if and only if $\mathit{PE}^{\sup}_{s_{\mathit{init}}}$ is finite. 
We will show how we can restrict ourselves to this case if there are no critical schedulers:

So, let $\mathcal{M}$ be an MDP with $\Pr^{\min}_{\mathcal{M},s_{\mathit{init}}}(\Diamond \mathit{goal})=0$ and suppose there are no critical schedulers for $\mathcal{M}$. Let $S_0$ be the set of all states reachable from $s_{\mathit{init}}$ while only choosing actions in $Act^{\min}$. As there are no critical schedulers, $(S_0,Act^{\min})$ does not contain positive cycles. So, there is a finite maximal weight $w_s$ among paths leading from $s_{\mathit{init}}$ to  $s$ in  $S_0$.
Consider the following MDP $\mathcal{N}$: It contains the MDP $\mathcal{M}$ and a new initial state $t_{\mathit{init}}$.
For each $s\in S_0$ and each $\alpha\in Act(s)\setminus Act^{\min}(s)$, $\mathcal{N}$ also contains a new state $t_{s,\alpha}$ which is reachable from $t_{\mathit{init}}$ via an action $\beta_{s,\alpha}$ with weight $w_s$ and probability $1$. 
In $t_{s,\alpha}$, only action $\alpha$ with the same probability distribution over successors and the same weight as in $s$ is enabled. So in $\cN$, one has to decide immediately in which state to leave $S_0$ and one accumulates the maximal weight which can be accumulated in $\cM$ to reach this state in $S_0$. In this way, we ensure that $\Pr^{\min}_{\mathcal{N}, t_{\mathit{init}}} (\Diamond \mathit{goal}) >0$.

\begin{restatable}{myprop}{reductionPosPr} \label{prop:reduction}
The constructed MDP $\mathcal{N}$ satisfies $ \mathit{CE}^{\sup}_{\mathcal{N},t_{\mathit{init}}}= \mathit{CE}^{\sup}_{\mathcal{M},s_{\mathit{init}}} $.
\end{restatable}

We can rely on this reduction to an MDP in which $\mathit{goal}$ is reached with positive probability for $\epsilon$-approximations and the exact computation of the optimal conditional expectation. In particular, the values $w_s$ for $s\in S_0$ are easy to compute by classical shortest path algorithms on weighted graphs. Furthermore, we can now decide the finiteness of the maximal conditional expectation.

\begin{restatable}{myprop}{finiteCE}
For an arbitrary MDP $\mathcal{M}$, $\mathit{CE}^{\sup}_{\mathcal{M},s_{\mathit{init}}}$ is finite if and only if there are no positively weight-divergent end components and no critical schedulers.
\end{restatable}

\tudparagraph{1ex}{Partial and conditional expectations in Markov chains.}
Markov chains with integer weights can be seen as MDPs with only one action $\alpha$ enabled in every state. Consequently, there is only one scheduler for a Markov chain. Hence, we drop the superscripts in $p^{\max}$ and $\mathit{PE}^{\sup}$.

\begin{myprop}
The partial and conditional expectation in a Markov chain $\mathcal{C}$ are computable in polynomial time.
\end{myprop}

\begin{proof} Let $\alpha$ be the only action available in $\mathcal{C}$. Assume that all states from which $\mathit{goal}$ is not reachable have been collapsed to an absorbing state $\fail$. Then  $\mathit{PE}_{\mathcal{C}, \sinit} $ is the value of $x_{\sinit}$ in the unique solution to the following system of linear equations with one variable $x_s$ for each state $s$:
\begin{eqnarray*} 
&&x_{\mathit{goal}}=x_{\fail}=0 ,\\
&&x_s  =  wgt(s,\alpha)\cdot p_s + \sum_{t} P(s,\alpha,t)\cdot x_t \text{ for }s\in S\setminus\{\mathit{goal},\fail\}.
\end{eqnarray*}
The existence of a unique solution follows from the fact that $\{\goal\}$ and $\{\fail\}$ are the only end components (see \cite{puterman1994}). It is straight-forward to check that $(\mathit{PE}_{\mathcal{C},s})_{s\in S}$ is this unique solution.
The conditional expectation is obtained from the partial expectation by dividing by the probability $p_{s_{\mathit{init}}}$ to reach the goal.
\end{proof}

This result can be seen as a special case of the following result. Restricting ourselves to schedulers which reach the goal with maximal or minimal probability in an MDP without positively weight-divergent end components, linear programming allows us to compute the following two memoryless deterministic schedulers (see \cite{chen2013,tacas2017}).
\begin{myprop}\label{Scheduler Max}
Let $\mathcal{M}$ be an MDP without positively weight-divergent end components. There is a  scheduler $\mathfrak{Max}\in \mathit{MD}^\mathcal{M}$ such that for each $s\in S$ we have
$\Pr^\mathfrak{Max}_s(\Diamond \mathit{goal})= p_s^{{\max}}$ and 
$\mathit{PE}^{\mathfrak{Max}}_s= \mathrm{sup}_\mathfrak{S} \mathit{PE}^{\mathfrak{S}}_s$ where the supremum is taken over all schedulers  $\mathfrak{S}$ with $\Pr^\mathfrak{S}_s(\Diamond \mathit{goal})= p_s^{{\max}}$.
Similarly, there is a  scheduler $\mathfrak{Min}\in \mathit{MD}^\mathcal{M}$ maximizing the partial expectation among all schedulers reaching the goal with minimal probability. Both these schedulers and their partial expectations are computable in polynomial time.
\end{myprop}

These schedulers will play a crucial role for the approximation of the maximal partial expectation and the exact computation of maximal partial expectations in MDPs with non-negative weights.


\tudparagraph{1ex.}{Partial expectations in MDPs with non-negative weights.}
In \cite{chen2013}, the computation of maximal partial expectations in stochastic multiplayer games with non-negative weights is presented. We adapt this approach to MDPs with non-negative weights.
A key result is the existence of a \emph{saturation point}, a bound on the accumulated weight above which optimal schedulers do not need memory.

In the sequel, let $R\in\mathbb{Q}$ be arbitrary, let  $\mathcal{M}$ be an MDP with non-negative weights, $\mathit{PE}^{\sup}_{s_{\mathit{init}}}<\infty$, and assume that end components have negative maximal mean payoff (see Proposition \ref{prop:removingzeroECs}). 
A saturation point for bias $R$ is a natural number $\mathfrak{p}$ such that there is a scheduler $\mathfrak{S}$ with  $\mathit{PE}^{\mathfrak{S}}_{s_{\mathit{init}}}[R]=\mathit{PE}^{{\sup}}_{s_{\mathit{init}}}[R]$ which is memoryless and deterministic as soon as the accumulated weight reaches $\mathfrak{p}$. I.e. for any two paths $\pi$ and $\pi^\prime$, with $last(\pi)=last(\pi^\prime)$ and $wgt(\pi),wgt(\pi^\prime)>\mathfrak{p}$, $\mathfrak{S}(\pi)=\mathfrak{S}(\pi^\prime)$. 

Transferring the idea behind the saturation point for conditional expectations given in \cite{tacas2017}, we  provide the following saturation point which can be considerably smaller than the saturation point given in \cite{chen2013} in stochastic multiplayer games. 
Detailed proofs to this section are given in Appendix \ref{app:non-negative}.

\begin{myprop} \label{SaturationPoint}
We define $p_{s,\alpha}^{\max}:= \sum_{t\in S} P(s,\alpha,t)\cdot p_t^{\max}$ and $\mathit{PE}^\mathfrak{Max}_{s,\alpha}:= p_{s,\alpha}^{\max}\cdot wgt(s,\alpha) + \sum_{t\in S} P(s,\alpha,t)\cdot  \mathit{PE}^\mathfrak{Max}_t$.
Then, \[\mathfrak{p}_R:= \mathrm{sup} \left\{\left.\frac{\mathit{PE}^\mathfrak{Max}_{s,\alpha}-\mathit{PE}^\mathfrak{Max}_{s}}{p^{\max}_s-p^{\max}_{s,\alpha}}\right|s\in S, \alpha\in Act(s)\setminus Act^{\max}(s)\right\}-R\] is an upper saturation point for bias $R$ in $\mathcal{M}$.
\end{myprop}

The saturation point $\mathfrak{p}_R$ is chosen such that, as soon as the accumulated weight exceeds $\mathfrak{p}_R$, the scheduler $\mathfrak{Max}$ is better than any scheduler deviating from $\mathfrak{Max}$ for only one step.
So, the proposition states  that $\mathfrak{Max}$ is then also better than any other scheduler.

As all values involved in the computation can be determined by linear programming, the saturation point $\mathfrak{p}_R$ is computable in polynomial time.
This also means that the logarithmic length of $\mathfrak{p}_R$ is polynomial in the size of $\mathcal{M}$ and hence $\mathfrak{p}_R$ itself is at most exponential in the size of $\mathcal{M}$.

\begin{restatable}{myprop}{computationPE}
Let $R\in\mathbb{Q}$ and let $B_R$ be the least integer greater or equal to $\mathfrak{p}_R+ \max_{s\in S,\alpha\in Act(s)} wgt(s,\alpha)$ and let $S^\prime := S\setminus\{\mathit{goal},fail\}$. The values $(\mathit{PE}^{\sup}_{s_{\mathit{init}}}[r{+}R])_{s\in S^\prime, 0\leq r\leq B_R}$ form the unique solution to the following linear program  in the variables $(x_{s,r})_{s\in S^\prime, 0\leq r\leq B_R}$ (r ranges over integers): 

Minimize $\sum_{s\in S^\prime,0\leq r\leq B_R}x_{s,r}$ under the following constraints:
\begin{eqnarray*}
&&\text{For }r \geq \mathfrak{p}_R:  x_{s,r}= p_s^{\max} \cdot (r{+}R) + E^\mathfrak{Max}_s, \\
&&\text{for }r < \mathfrak{p}_R \text{ and }\alpha\in Act(s) :\\
&&x_{s,r}\geq P(s,\alpha,\mathit{goal}) \cdot (r{+}R{+}\mathit{wgt}(s,\alpha)) + \sum_{t\in S^\prime} P(s,\alpha,t) \cdot x_{t,r{+}\mathit{wgt}(s,\alpha).}
\end{eqnarray*}
\end{restatable}
From a solution $x$ to the linear program, we can easily extract an optimal weight-based deterministic scheduler. This scheduler only needs finite memory because the accumulated weight increases monotonically along paths and as soon as the saturation point is reached $\mathfrak{Max}$ provides the optimal decisions.
 As $B_R$ is exponential in the size of $\mathcal{M}$, the computation of the optimal partial expectation via this linear program runs in time exponential in the size of $\mathcal{M}$.



\section{Existence of Optimal Schedulers} \label{sec:optimalWD}


We prove that there are optimal weight-based deterministic schedulers for partial and conditional expectations. After showing that, if finite, $\mathit{PE}^{\sup}_{s_{\mathit{init}}}$ is equal to $\sup_{\mathfrak{S}\in \mathit{WD}^{\mathcal{M}}} \mathit{PE}^\mathfrak{S}_{s_{\mathit{init}}}$, we take an analytic approach to show that there is indeed a weight-based deterministic scheduler maximizing the partial expectation. We define a metric on  $\mathit{WD}^\mathcal{M}$ turning it into a compact space. Then, we prove that the function assigning the partial expectation to schedulers is upper semi-continuous. We conclude that there is a weight-based deterministic scheduler obtaining the maximum.
Proofs to this section can be found in Appendix \ref{app:optimalWD}.

\begin{restatable}{myprop} {WDscheduler} \label{prop:supWD}
Let  $\mathcal{M}$ be an MDP with $\mathit{PE}^{\sup}_{s_{\mathit{init}}}<\infty$. Then we have  $\mathit{PE}^{\sup}_{s_{\mathit{init}}}=\sup_{\mathfrak{S}\in \mathit{WD}^\mathcal{M}} \mathit{PE}^{\mathfrak{S}}_{s_{\mathit{init}}}$.
\end{restatable}

\begin{proof}[Proof sketch] We can assume that all end components have negative maximal expected mean payoff (see Proposition \ref{prop:removingzeroECs}).
Given a scheduler $\mathfrak{S}\in \mathit{HR}^\mathcal{M}$, we take the expected number of times $\theta_{s,w}$ that $s$ is visited with accumulated weight $w$ under $\mathfrak{S}$ for each state-weight pair $(s,w)$, and the expected number of times $\theta_{s,w,\alpha}$ that  $\mathfrak{S}$ then chooses $\alpha$. These values are finite due to the negative maximal mean payoff in end components. We define the scheduler $\mathfrak{T}\in \mathit{WR}^\mathcal{M}$ choosing $\alpha$ in $s$ with probability $\theta_{s,w,\alpha}/\theta_{s,w}$ when weight $w$ has been accumulated. Then, we show by standard arguments that we can replace all probability distributions that $\mathfrak{T}$ chooses by Dirac distributions to obtain a scheduler $\mathfrak{T}^\prime\in \mathit{WD}^\mathcal{M}$ such that $\mathit{PE}^{\mathfrak{T}^\prime}_{s_{\mathit{init}}}\geq \mathit{PE}^{\mathfrak{S}}_{s_{\mathit{init}}}$.
\end{proof}

It remains to show that the supremum is obtained by a weight-based deterministic scheduler.
Given an MDP $\mathcal{M}$ with arbitrary integer weights, we define the following metric $d^\mathcal{M}$ on the set of weight-based deterministic schedulers, i.e. on the set of functions from $S\times \mathbb{Z} \to Act$:
For two such schedulers $\mathfrak{S}$ and $\mathfrak{T}$, we let  $d^\mathcal{M}(\mathfrak{S},\mathfrak{T}) := 2^{-R}$ where $R$ is the greatest natural number such that $\mathfrak{S}\upharpoonright S\times \{-(R{-}1), \dots , R{-}1\} = \mathfrak{T}~\upharpoonright~S\times \{-(R{-}1), \dots , R{-}1\}$ or $\infty$ if there is no greatest such natural number.

\begin{mylem}
The metric space $(Act^{S\times \mathbb{Z}},d^\mathcal{M})$ is compact.
\end{mylem}

Having defined this compact space of schedulers, we can  rely on the analytic notion of upper semi-continuity.

\begin{mylem}[Upper Semi-Continuity of Partial Expectations]
If $\mathit{PE}^{\sup}_{s_{\mathit{init}}}$ is finite in $\mathcal{M}$, then the function $\mathit{PE}: (\mathit{WD}, d^{\mathit{WD}}) \to (\mathbb{R}_\infty,d^{euclid})$ assigning $\mathit{PE}^\mathfrak{S}_{s_{\mathit{init}}}$ to a weight-based deterministic scheduler $\mathfrak{S}$ is upper semi-continuous. 
\end{mylem}

The technical proof of this lemma can be found in Appendix \ref{app:optimalWD}. We arrive at the main result of this section.

\begin{mythm}[Existence of Optimal Schedulers for Partial Expectations]
If $\mathit{PE}^{\sup}_{s_{\mathit{init}}}$ is finite in an MDP $\mathcal{M}$, then there is a weight-based deterministic scheduler $\mathfrak{S}$ with 
$\mathit{PE}^{\sup}_{s_{\mathit{init}}}=\mathit{PE}^\mathfrak{S}_{s_{\mathit{init}}}$.
\end{mythm}

\begin{proof}
If $\mathit{PE}^{\sup}_{s_{\mathit{init}}}$ is finite, then the map $\mathit{PE}: (\mathit{WD}, d^{\mathit{WD}}) \to (\mathbb{R}_\infty,d^{euclid})$ is upper semi-continuous. So, this map has a maximum because $(\mathit{WD},d^{\mathit{WD}})$ is a compact metric space. 
\end{proof}

\begin{mycor}[Existence of Optimal Schedulers for Conditional Expectations]
If $\mathit{CE}^{\sup}_{s_{\mathit{init}}}$ is finite in an MDP $\mathcal{M}$, then there is a weight-based deterministic scheduler $\mathfrak{S}$ with 
$\mathit{CE}^{\sup}_{s_{\mathit{init}}}=\mathit{CE}^\mathfrak{S}_{s_{\mathit{init}}}$.
\end{mycor}

\begin{proof}
By  Proposition \ref{prop:reduction}, we can assume that $\Pr^{\min}_{s_{\mathit{init}}}(\Diamond \mathit{goal})>0$. We know that $\mathit{PE}^{\sup}_{s_{\mathit{init}}}[-\mathit{CE}^{\sup}_{s_{\mathit{init}}}]=0$ and that there is a weight-based deterministic scheduler $\mathfrak{S}$ with $\mathit{PE}^{\mathfrak{S}}_{s_{\mathit{init}}}[-\mathit{CE}^{\sup}_{s_{\mathit{init}}}]=0$. By Proposition \ref{prop:threshold}, $\mathfrak{S}$ maximizes the conditional expectation as it reaches $\mathit{goal}$ with positive probability.
\end{proof}

In MDPs with non-negative weights, the optimal decision in a state $s$ only depends on $s$ as soon as the accumulated weight exceeds a saturation point. In MDPs with arbitrary integer weights, it is possible that the optimal choice of action does not become stable for increasing values of accumulated weight as we see in the following example.

\FloatBarrier

\begin{figure} 
\begin{center}
\resizebox{0.75\textwidth}{!}{%
\begin{tikzpicture}[scale=1,auto,node distance=8mm,>=latex,font=\small]
    \tikzstyle{round}=[thick,draw=black,circle]

    \node[round, draw=black,minimum size=0pt,minimum size=20pt] (sinit) {$s_{\mathit{init}}$};
    \node[round,right=13mm of sinit,minimum size=20pt] (s3) {$t$};
        
    \node[round,right=13mm of s3,minimum size=20pt] (goal) {$\mathit{goal}$};
    \node[round,below=15mm of goal,minimum size=20pt] (fail) {$\fail$};
    
    \node[above=18mm of s3] (text) {The MDP $\mathcal{N}$:};

 \draw[color=black ,->] (sinit) [out=40,in=100,loop,min distance=20mm] edge  node [very near start, anchor=center] (n5) {} node [pos=0.6,above] {} (sinit) ;
    \draw[color=black ,->] (sinit) edge node [near start, anchor=center] (n6) {} node [pos=0.7,above] {} (s3) ;
    \draw[color=black , very thick] (n5.center) edge [bend left=35] node [pos=0.3] {$\tau | +1$} (n6.center);

    \draw[color=blue ,->] (s3) [out=40,in=100,loop,min distance=20mm] edge  node [very near start, anchor=center] (n5) {} node [pos=0.6,above] {} (s3) ;
    \draw[color=blue ,->] (s3) edge node [near start, anchor=center] (n6) {} node [pos=0.7,above] {} (goal) ;
    \draw[color=blue , very thick] (n5.center) edge [bend left=35] node [pos=0.3] {$\beta | -1$} (n6.center);
    
    \draw[color=red ,->] (s3) edge [bend right= 30]  node [ anchor=center] (m5) {} node [pos=0.7,above] {} (goal) ;
    \draw[color=red ,->] (s3) edge node [near start, anchor=center] (m6) {} node [pos=0.7,above] {} (fail) ;
 \draw[color=red , very thick] (m5.center) edge [bend left=35] node [pos=0.3] {$\alpha | 0$} (m6.center);

\end{tikzpicture}
\hspace{40pt}
\begin{tikzpicture}[scale=1,auto,node distance=8mm,>=latex,font=\small]
    \tikzstyle{round}=[thick,draw=black,circle]

    \node[round, draw=black,minimum size=0pt,minimum size=10pt] (sinit) {$s_{\mathit{init}}$};
    \node[round,right=15mm of sinit,minimum size=20pt] (s3) {$t$};
    
    \node[round,above=18mm of s3,minimum size=20pt] (q0) {$q_0$};
     \node[round,right=15mm of q0,minimum size=20pt] (q1) {$q_1$};

    \node[round,above=15mm of q1,minimum size=20pt] (goal2) {$\mathit{goal}$};
    \node[round,below=16mm of q1,minimum size=20pt] (fail2) {$\fail$};
    
     \node[left=20mm of goal2] (text) {The MDP $\mathcal{M}$:};

 \draw[color=black ,->] (sinit) [out=40,in=100,loop,min distance=20mm] edge  node [very near start, anchor=center] (n5) {} node [pos=0.6,above] {} (sinit) ;
    \draw[color=black ,->] (sinit) edge node [near start, anchor=center] (n6) {} node [pos=0.7,above] {} (s3) ;
    \draw[color=black , very thick] (n5.center) edge [bend left=35] node [pos=0.3] {$\tau | +1$} (n6.center);

 \draw[color=darkgray,->] (s3)  edge  node [pos=0.4,right] {$\gamma / 0$} (q0) ;
 \draw[color=purple,->] (s3)  edge  node [pos=0.4,right] {$\delta / 0$} (q1) ;

        \draw[color=red,->] (q1)  edge node [near start, anchor=center] (l3) {} node [pos=0.6,below] {} (goal2) ;
    \draw[color=red,->] (q1) edge node [  near start,  anchor=center] (l4) {} node [pos=0.6,left] {}(fail2) ;
    \draw[color=red, very thick] (l3.center) [bend left=90] edge node [pos=0.5] {$\alpha | 0$}  (l4.center);
    
    \draw[color=blue ,->] (q0) [out=220,in=190,loop,min distance=20mm] edge  node [very near start, anchor=center] (m5) {} node [pos=0.6,above] {} (goal2) ;
    \draw[color=blue ,->] (q0) [bend right=15] edge node [near start, anchor=center] (m6) {} node [pos=0.7,above] {} (q1) ;
    \draw[color=blue , very thick] (m6.center) edge [bend left=35] node [pos=0.3] {$\beta | -1$} (m5.center);
    
      \draw[color=blue ,->] (q1) [out=140,in=250,loop,min distance=10mm] edge  node [very near start, anchor=center] (m7) {} node [pos=0.6,above] {} (goal2) ;
    \draw[color=blue ,->] (q1) [bend right=15] edge node [near start, anchor=center] (m8) {} node [pos=0.7,above] {} (q0) ;
    \draw[color=blue , very thick] (m8.center) edge [bend left=35] node [pos=0.3] {$\beta | -1$} (m7.center);
\end{tikzpicture}
}
\end{center}%
\caption{All non-trivial transition probabilities are $1/2$. In the MDP $\mathcal{M}$, the optimal choice to maximize the partial expectation in $t$ depends on the parity of the accumulated weight.}
\label{fig:nonstationary}
\end{figure}
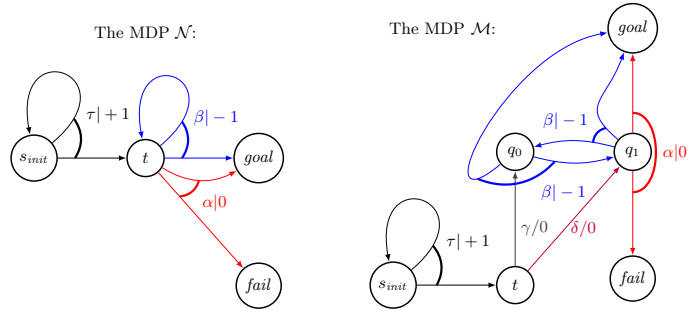

\begin{myex}
Let us first consider the MDP $\mathcal{N}$ depicted in Figure \ref{fig:nonstationary}. 
Let $\pi$ be a path reaching $t$ for the first time with accumulated weight $r$. Consider a scheduler which chooses $\beta$ for the first $k$ times and then $\alpha$. In this situation, the partial expectation from this point on is:
\[
\frac{1}{2^{k+1}}\left(  r{-}k  \right) + \sum_{i=1}^k \frac{1}{2^i} (r{-}i)= \frac{1}{2^{k+1}}+\sum_{i=1}^{k+1} \frac{1}{2^i} (r{-}i) = \frac {k{-}r+4}{2^{k+1}}+r{-}2.
\] 
For $r\geq2$, this partial expectation has its unique maximum for the choice $k=r{-}2$. This already shows that an optimal scheduler needs infinite memory. No matter how much weight $r$ has been accumulated when reaching $t$, the optimal scheduler has to count the $r{-}2$ times it chooses $\beta$.

Furthermore, we can transfer the optimal scheduler for the MDP $\mathcal{N}$ to the MDP $\mathcal{M}$. In state $t$, we have to make a nondeterministic choice between two action leading to the states $q_0$ and $q_1$, respectively. In both of these states, action $\beta$ is enabled which  behaves like the same action in the MDP $\mathcal{N}$ except that it moves between the two states if $\mathit{goal}$ is not reached. So, the action $\alpha$ is only enabled every other step. As in $\mathcal{N}$,  we want to choose $\alpha$ after choosing $\beta$ $r{-}2$ times if we arrived in $t$ with accumulated weight $r\geq 2$. So, the choice in $t$ depends on the parity of $r$: For $r=1$ or $r$ even, we choose $\delta$. For odd $r\geq 3$, we choose $\gamma$. This shows that the optimal scheduler in the MDP $\mathcal{M}$  needs  specific information about the accumulated weight, in this case the parity, no matter how much weight has been accumulated.

In the example, the optimal scheduler has a periodic behavior when fixing a state and looking at optimal decisions for increasing values of accumulated weight. The question whether an optimal scheduler always has such a periodic behavior remains open.
\end{myex}

\section{Approximation} \label{sec:approximation}


As the optimal values for partial and conditional expectation can be irrational, there is no hope to compute these values by linear programming as in the case of non-negative weights. In this section, we show how we can nevertheless approximate the values. The main result is the following.

\begin{mythm} \label{thm:approximation}
Let $\mathcal{M}$ be an MDP with $\mathit{PE}^{\sup}_{\mathcal{M},s_{\mathit{init}}}<\infty$ and $\epsilon>0$.
The maximal partial expectation $\mathit{PE}^{\sup}_{\mathcal{M},s_{\mathit{init}}}$ can be approximated up to an absolute error of $\epsilon$ in time exponential in the size of $\mathcal{M}$ and polynomial in $\log(1/\epsilon)$.
If further, $\mathit{CE}^{\sup}_{\mathcal{M},s_{\mathit{init}}}<\infty$, also $\mathit{CE}^{\sup}_{\mathcal{M},s_{\mathit{init}}}$ can be approximated up to an absolute error of $\epsilon$ in time exponential in the size of $\mathcal{M}$ and polynomial in $\log(1/\epsilon)$.
\end{mythm}
We first prove that upper bounds for  $\mathit{PE}^{\sup}_{\mathcal{M},s_{\mathit{init}}}$  and  $\mathit{CE}^{\sup}_{\mathcal{M},s_{\mathit{init}}}$ can be computed in polynomial time. Then, we show that there are $\epsilon$-optimal schedulers for the partial expectation which become memoryless as soon as the accumulated weight leaves a sufficiently large weight window around $0$. We compute the optimal partial expectation of such a scheduler  by linear programming. The result can then be extended to conditional expectations.

\tudparagraph{1ex}{Upper Bounds.} Let $\mathcal{M}$ be an MDP in which all end components have negative maximal mean payoff. 
Let $\delta$ be the minimal non-zero transition probability in $\mathcal{M}$ and $W:= \max_{s\in S, \alpha\in Act(s)} |\wgt(s,\alpha)|$. Moving through the MEC-quotient, the probability to reach an accumulated weight of $|S|\cdot W$ is bounded by $1-\delta^{|S|}$ as $\mathit{goal}$ or $\fail$ is reached within $S$ steps with probability at least $1-\delta^{|S|}$. It remains to show similar bounds inside an end component. 

We will use the characterization of the maximal mean payoff in terms of super-harmonic vectors due to Hordijk and Kallenberg \cite{hordijk1979} 
to define a supermartingale controlling the growth of the accumulated weight in an end component under any scheduler.
As the value vector for the maximal mean payoff in an end component is constant and negative in our case, the results of \cite{hordijk1979} yield:

\begin{myprop}[Hordijk, Kallenberg] \label{prop:MP}
Let $\mathcal{E}=(S,Act)$ be an end component with maximal mean payoff $-t$ for some $t>0$. Then there is a vector $(u_s)_{s\in S}$ such that
$ -t + u_s \geq wgt(s,\alpha) + \sum_{s^\prime\in S} P(s,\alpha,s^\prime) \cdot u_{s^\prime}$.

Furthermore, let $v$ be the vector (-t,\dots,-t) in $\mathbb{R}^S$. Then, $(v,u)$ is the solution to a linear program with $2|S|$ variables, $2|S||Act|$ inequalities, and coefficients formed from the transition probabilities and weights in $\mathcal{E}$.
\end{myprop}

We will call the vector $u$ a \emph{super-potential} because the expected accumulated weight after $i$ steps is at most $u_s -\min_{t \in S} u_t - i\cdot t$ when starting in state $s$.
Let $\mathfrak{S}$ be a scheduler for $\mathcal{E}$ starting in some state $s$. We define the following random variables on $\mathfrak{S}$-runs in $\mathcal{E}$: let $s(i)\in S$ be the state after $i$ steps, let $\alpha(i)$ be the action chosen after $i$ steps, let $w(i)$ be the accumulated weight after $i$ steps, and let $\pi(i)$ be the history, i.e. the finite path after $i$ steps. 

\begin{mylem}
The sequence $m(i):=w(i)+u_{s(i)} $ satisfies  $\mathbb{E}(m(i+1)|\pi(0),\dots,\pi(i))$ $\leq m(i)-t$ for all $i$.%
\footnote{This means that $m(i)+i\cdot t$ is a super-martingale with respect to the history $\pi(i)$.}
\end{mylem}

\begin{proof}
By  Proposition \ref{prop:MP},
$\mathbb{E}(m(i+1)|\pi(0),\dots,\pi(i))-m(i) 
= wgt(s(i),\mathfrak{S}(\pi(i))) + \sum_{s^\prime\in S} P(s(i),\mathfrak{S}(\pi(i)),s^\prime)\cdot u_{s^\prime}  -u_{s(i)} 
\leq -t $.
 \end{proof}

We are going to  apply the following theorem by Blackwell \cite{blackwell1954}.

\begin{mythm}[Blackwell  \cite{blackwell1954}] \label{thm:blackwell}
Let $X_1, X_2, \dots $ be random variables, and let $S_n:= \sum_{k=1}^n X_k$. Assume that $|X_i|\leq 1$ for all $i$ and that there is a $u>0$ such that $\mathbb{E}(X_{n+1}|X_1,\dots,X_n) \leq -u$. Then,
$\Pr(\sup_{n\in\mathbb{N}} S_n \geq t) \leq \left( \frac{1-u}{1+u}\right)^t$.
\end{mythm}

We denote $\max_{s^\prime \in S} u_{s^\prime} - \min_{s^\prime \in S} u_{s^\prime}$ by $\|u\|$. Observe that $|m(i+1)-m(i)|\leq \|u\| +W=:c_{\mathcal{E}}$. We can rescale the sequence $m(i)$ by defining $m^\prime(i) := (m(i)-m(0))/c_\mathcal{E}$. This ensures that $m^\prime(0)=0$, $|m^\prime(i+1)-m^\prime(i)|\leq 1$ and $\mathbb{E}(m^\prime(i+1)|m^\prime(0),\dots,m^\prime(i))\leq -t/c_\mathcal{E}$ for all $i$. In this way, we arrive at the following conclusion, putting $\lambda_{\mathcal{E}}:=\frac{1-t/c_{\mathcal{E}}}{1+t/c_{\mathcal{E}}}$.

\begin{mycor} \label{cor:boundEC} For any scheduler $\mathfrak{S}$ and any starting state $s$ in $\mathcal{E}$, we have
$\Pr^\mathfrak{S}_s (\Diamond wgt \geq  (k{+}1)\cdot c_{\mathcal{E}})  \leq \lambda_{\mathcal{E}}^{k}$.
\end{mycor}

\begin{proof} By Theorem \ref{thm:blackwell},
$\Pr^\mathfrak{S}_s (\Diamond wgt \geq (k+1)\cdot c_{\mathcal{E}} ) \leq  \Pr^\mathfrak{S}_s (\Diamond wgt \geq \|u\| + k\cdot c_{\mathcal{E}}) 
\leq \Pr^\mathfrak{S}_s (\exists i : m(i)-m(0)\geq k\cdot c_{\mathcal{E}}) 
 = \Pr^\mathfrak{S}_s (\sup_{i\in\mathbb{N}} m^\prime(i)\geq k) 
  \leq \left(\frac{1-t/c_{\mathcal{E}}}{1+t/c_{\mathcal{E}}}\right)^{k}. $
\end{proof}

Let  $\mathit{MEC}$ be the set of maximal end components in $\mathcal{M}$. For each $\mathcal{E}\in \mathit{MEC}$, let $\lambda_\mathcal{E}$ and $c_\mathcal{E}$ be as in Corollary \ref{cor:boundEC}. 
Define $\lambda_{\mathcal{M}} := 1- (\delta^{|S|}\cdot \prod_{\mathcal{E}\in \mathit{MEC}} (1- \lambda_\mathcal{E}))$, and $c_{\mathcal{M}}:=|S|\cdot W + \sum_{\mathcal{E}\in \mathit{MEC}} c_\mathcal{E}$. Then an accumulated weight of $c_\mathcal{M}$ cannot be reached with a probability greater than $\lambda_\mathcal{M}$ because reaching accumulated weight $c_\mathcal{M}$ would require reaching weight $c_\mathcal{E}$ in some end component $\mathcal{E}$ or reaching weight $|S|\cdot W$ in the MEC-quotient and $1-\lambda_\mathcal{M}$ is a lower bound on the probability that none of this happens (under any scheduler).

\begin{myprop} Let $\mathcal{M}$ be an MDP with $\mathit{PE}^{\sup}_{s_{\mathit{init}}}<\infty$. 
There is an upper bound $\mathit{PE}^{ub}$ for the partial expectation in $\mathcal{M}$ computable in polynomial time.
\end{myprop}

\begin{proof}
In any end component $\mathcal{E}$, the maximal mean payoff $-t$ and the super-potential $u$ are computable in polynomial time. Hence, $c_\mathcal{E}$ and $\lambda_\mathcal{E}$, and in turn also $c_\mathcal{M}$ and $\lambda_\mathcal{M}$ are also computable in polynomial time. 
When we reach accumulated weight $c_\mathcal{M}$ for the first time, the actual accumulated weight is at most $c_\mathcal{M}+W$. So, we conclude that $\Pr^{\max}_s(\Diamond wgt\geq k\cdot(c_\mathcal{M}+W))\leq \lambda_\mathcal{M}^k$.
The partial expectation can now be bounded by
$\sum_{k=0}^\infty (k+1)\cdot (c_\mathcal{M}+W)\cdot \lambda_\mathcal{M}^k= \frac{c_\mathcal{M}+W}{(1-\lambda_\mathcal{M})^2}$. 
\end{proof}

\begin{mycor} Let $\mathcal{M}$ be an MDP with $\mathit{CE}^{\sup}_{\mathcal{M},s_{\mathit{init}}}<\infty$. 
There is an upper bound $\mathit{CE}^{ub}$ for the conditional expectation in $\mathcal{M}$ computable in polynomial time.
\end{mycor}

\begin{proof}
By Proposition \ref{prop:reduction}, we can construct an MDP $\mathcal{N}$ in which $\mathit{goal}$ is reached with probability $q>0$ in polynomial time with $\mathit{CE}^{\sup}_{\mathcal{M},s_{\mathit{init}}}=\mathit{CE}^{\sup}_{\mathcal{N},s_{\mathit{init}}}$.  Now, $\mathit{CE}^{ub}:= \mathit{PE}^{ub}/q$ is an upper bound for the conditional expectation in $\mathcal{M}$. 
\end{proof}


\tudparagraph{1ex}{Approximating optimal partial expectations.}
The idea for the approximation is to assume that the partial expectation is $\mathit{PE}^{\mathfrak{Max}}_{s_{\mathit{init}}}+w\cdot p^{\max}_s$ if a high weight $w$ has been accumulated in state $s$. Similarly, for small weights $w^\prime$, we use the value $\mathit{PE}^{\mathfrak{Min}}_{s_{\mathit{init}}}+w\cdot p^{\min}_s$.
We will first provide a lower ``saturation point" making sure that only actions minimizing the probability to reach the goal are used by an optimal scheduler as soon as the accumulated weight drops below this saturation point. Proofs to this section can be found in Appendix \ref{app:approximationPE}

\begin{myprop} Let $\mathcal{M}$ be an MDP with $\mathit{PE}^{\sup}_{s_{\mathit{init}}}<\infty$. 
Let $s\in S$ and let $\mathfrak{q}_s:=\frac{\mathit{PE}^{ub}-\mathit{PE}^{\mathfrak{Min}}_s}{p_s^{\min} - \min\limits_{\alpha\not\in Act^{\min}(s)} p_{s,\alpha}^{\min}}$. Then any weight-based deterministic scheduler $\mathfrak{S}$ maximizing the partial expectation in $\mathcal{M}$ satisfies $\mathfrak{S}(s,w)\in Act^{\min}(s)$ if $w\leq \mathfrak{q}_s$.
\end{myprop}

Let $\mathfrak{q}:=\min_{s\in S} \mathfrak{q}_s$ and let $D:=\mathit{PE}^{ub}-\min \{ \mathit{PE}^\mathfrak{Max}_s , \mathit{PE}^\mathfrak{Min}_s |s\in S\}$.
Given $\epsilon>0$, we define $R^+_\epsilon := (c_\mathcal{M}+W)\cdot \left\lceil \frac{\log(2D)+\log(1/\epsilon)}{\log(1/\lambda_\mathcal{M})}\right\rceil$ and $R^-_\epsilon:= \mathfrak{q}-R^+_\epsilon$.

\begin{mythm}
There is a weight-based deterministic scheduler $\mathfrak{S}$ such that the scheduler $\mathfrak{T}$ defined by
\[\mathfrak{T} (\pi) = 
\begin{cases}
\mathfrak{S}(\pi) & \text{if any prefix }\pi^\prime\text{ of }\pi\text{ satisfies }R^-_\epsilon\leq wgt(\pi^\prime)\leq R^+_\epsilon,\\
\mathfrak{Max}(\pi) & \text{if the shortest prefix }\pi^\prime\text{ of }\pi\text{ with } wgt(\pi^\prime)\not\in [R^-_\epsilon,R^+_\epsilon] \\
& \text{satisfies }wgt(\pi^\prime)>R^+_\epsilon, \\
\mathfrak{Min}(\pi) & \text{otherwise,}
\end{cases}\]
satisfies $\mathit{PE}^\mathfrak{T}_{s_{\mathit{init}}}\geq \mathit{PE}^{\sup}_{s_{\mathit{init}}} - \epsilon$. 
\end{mythm}

This result now allows us to compute an $\epsilon$-approximation and an $\epsilon$-optimal scheduler with finite memory by linear programming, 
similar to the case of non-negative weights, 
in a linear program with $R^+_\epsilon + R^-_\epsilon$ many variables and $|Act|$-times as many inequalities. 

\begin{mycor}
 $\mathit{PE}^{\sup}_{s_{\mathit{init}}}$ can be approximated up to an absolute error of $\epsilon$ in time exponential in the size of $\mathcal{M}$ and polynomial in $\log(1/\epsilon)$.
\end{mycor}

If the logarithmic length of $\theta\in\mathbb{Q}$ is polynomial in the size of $\mathcal{M}$, we can also approximate $\mathit{PE}^{\sup}_{s_{\mathit{init}}}[\theta]$ up to an absolute error of $\epsilon$ in time exponential in the size of $\mathcal{M}$ and polynomial in $\log(1/\epsilon)$: We can  add a new initial state $s$ with a transition to $s_{\mathit{init}}$ with weight $\theta$ and approximate $\mathit{PE}^{\sup}_{s}$ in the new MDP.


\tudparagraph{1ex}{Transfer to conditional expectations.}
Let $\mathcal{M}$ be an MDP with $\mathit{CE}^{sup}_{s_{\mathit{init}}}<\infty$ and $\epsilon>0$. By Proposition \ref{prop:reduction}, we can assume that $\Pr^{\min}_{\mathcal{M},s_{\mathit{init}}}(\Diamond \mathit{goal})=:p$ is positive. Clearly, $\mathit{CE}^{sup}_{s_{\mathit{init}}}\in [\mathit{CE}^{\mathfrak{Max}}_{s_{\mathit{init}}},\mathit{CE}^{ub}]$. We perform a binary search to approximate $\mathit{CE}^{sup}_{s_{\mathit{init}}}$:
We put $A_0:=\mathit{CE}^{\mathfrak{Max}}_{s_{\mathit{init}}}$ and $B_0:=\mathit{CE}^{ub}$. Given $A_i$ and $B_i$, let $\theta_i := (A_i+B_i)/2$. Then, we approximate $\mathit{PE}^{\sup}_{s_{\mathit{init}}}[-\theta_i]$ up to an absolute error of $p\cdot \epsilon$. 
 Let $E_i$ be the value of this approximation. If $E_i \in [-2p\cdot \epsilon, 2p\cdot \epsilon]$, terminate and return $\theta_i$ as the approximation for $\mathit{CE}^{\sup}_{s_{\mathit{init}}}$. If $E_i<-2p\cdot \epsilon$, put $A_{i+1}:=A_i$ and $B_{i+1}:=\theta_i$, and repeat. If $E_i>2p\cdot\epsilon$, put $A_{i+1}:=\theta_i$ and $B_{i+1}:= B_i$, and repeat.

\begin{myprop}
The procedure terminates after at most $\lceil \log((A_0-B_0)/(p\cdot\epsilon))\rceil$ iterations and returns an $3\epsilon$-approximation of $\mathit{CE}^{\sup}_{s_{\mathit{init}}}$ in time  exponential in the size of $\mathcal{M}$ and polynomial in $\log(1/\epsilon)$.
\end{myprop}
\noindent The proof can be found in Appendix \ref{app:approximationCE}. This finishes the proof of Theorem \ref{thm:approximation}.


\section{Conclusion} \label{sec:conclusion}

Compared to the setting of non-negative weights, the optimization of partial and conditional expectations faces substantial new difficulties in the setting of integer weights. The optimal values can be irrational showing that the linear programming approaches from the setting of non-negative weights cannot be applied for the computation of optimal values. We showed that this approach can nevertheless be adapted for approximation algorithms.
Further, we were able to show that  there are  optimal weight-based deterministic schedulers. These schedulers, however, can require infinite memory and it remains open whether we can further restrict the class of schedulers necessary for the optimization. In examples, we have seen that optimal schedulers can switch periodically between actions they choose for increasing values of accumulated weight. 
Further insights on the behavior of optimal schedulers would be helpful to address threshold problems (``Is $\mathit{PE}^{\sup}_{s_{\mathit{init}}} \geq \theta$?'').


\bibliographystyle{abbrv}
\bibliography{references}


\chapter*{Appendix}
\appendix
\section{Partial and Conditional Expectations in Markov Decision Processes}


In this section, we give  proofs to the claims of Section \ref{sec:PE}.

\subsection{Finiteness and Preprocessing} \label{app:finiteness}


The finiteness of maximal partial expectations depends on the existence of positively weight-divergent end components.
Using the construction from \cite{lics2018} which removes end components only containing $0$-weight cycles, we can show the following:
\removingzeroECs*

\begin{proof}
In an end component which has non-negative maximal expected mean payoff and which is not positively weight-divergent, all cycles have weight $0$ (see \cite{lics2018}). 
We will use the so-called spider construction from \cite{lics2018} with a small modification to remove such end components: 
So, let $\cM$ be an MDP and let $\cE=(E,\ActEC)$ be an end component of $\cM$ in which all cycles have weight $0$.
The spider construction successively flattens sub-end components which contain exactly one action per state:
So, let $E^\prime \subseteq E$ and for each $s\in E^\prime$, let $\alpha_s\in \ActEC(s)$ 
such that $\cE^\prime = \{(s,\alpha_s) | s\in E^\prime\}$ is an end component.
We pick a state $s_0\in E^\prime$. As all cycles in $\cE^\prime$ have weight $0$, 
there is a unique weight $w_s$ for each $s\in E^\prime$ such that all paths from $s$ to $s_0$ in $\cE^\prime$ have weight $w_s$.
The spider construction now does the following:
\begin{enumerate}
\item Disable the action $\alpha_s$ in $s$ for all $s\in E^\prime$.
\item Enable a new action $\tau$ in $s$ with $P(s,\tau,s_0)=1$ and $\wgt(s,\tau)=w_s$ for all $s\not = s_0$ in $E^\prime$.
\item For each state $s\not = s_0$ in $E^\prime$ and each action $\beta \in \Act_\cM (s)\setminus \{\alpha_s\}$, disable $\beta$ in $s$ and 
instead enable a new action $\beta_s$ in $s_0$ with $P(s_0,\beta_s, t) = P(s,\beta,t)$ for all states $t$ in $\cM$ and $\wgt(s_0,\beta_s) = \wgt(s,\beta) - w_s $.
\end{enumerate}
We extend the construction by adding an absorbing state $\fail$ and additionally enabling a new action $\tau$ in $s_0$ with $P(s_0,\tau,\fail)=1$ and $\wgt(s_0,\tau)=0$.
We call the resulting MDP after one application of the construction $\cN^\prime$. In \cite{lics2018}, 
it is shown that polynomially many applications (in polynomial time in total)
 of the construction result in an MDP $\cN$ satisfying the first requirement in the statement.

Hence, it is sufficient to show the correspondence between schedulers claimed 
in the second requirement for the MDPs $\cM$ and $\cN^\prime$.
Given a scheduler $\sched$ for $\cM$, we construct the following scheduler $\sched^\prime$ for $\cN^\prime$: 
Whenever a run in $\cM$ under $\sched$ reaches $\cE^\prime$, 
let $p_\fail$ be the probability that $\sched$ will never leave $\cE^\prime$ again. 
Further, for each state $s$ in $\cE^\prime$ and each action $\beta\in Act(s)$ not belonging to $\cE$, 
let $p_{s,\beta}$ be the probability that $\sched$ leaves $\cE^\prime$ from $s$ via $\beta$. 
This behavior can now be mimicked in $\cN^\prime$: $\sched^\prime$ goes to $\fail$ 
with probability $p_\fail$ and takes the action $\beta_s$ in $s_0$ with probability $p_{s,\beta}$. 
It is straightforward to check that this does not affect the partial expectation or the probability to reach $\goal$.

Conversely, a scheduler $\sched^\prime$ for $\cN^\prime$ can easily be transformed to a scheduler $\sched$ for $\cM$: Whenever $\sched^\prime$ moves to $\fail$ from $s_0$, 
the scheduler $\sched$ stays in $\cE^\prime$ forever. If $\sched^\prime$ chooses $\beta_s$ in $s_0$, $\sched$ moves through $\cE^\prime$ until it reaches $s$. This happens almost surely. Then, $\sched$ chooses $\beta$. Again, it is easy to check that the partial expectation and the probability to reach $\goal$ are preserved.
\end{proof}

\finitenessPE*

\begin{proof}
Suppose there is a positively weight-divergent end component $\mathcal{E}$. 
Since $\mathcal{E}$ is reachable and we can accumulated arbitrarily high weights inside $\mathcal{E}$ with probability $1$, 
we can easily construct a sequence of schedulers whose partial expectation diverges to $+\infty$ by letting the schedulers stay 
in a positively weight divergent end component until an arbitrarily high weight has been accumulated, before they try to reach the goal.

Now, suppose that there are no positively weight-divergent end components. 
We can assume that all end components have negative maximal mean payoff (see Proposition \ref{prop:removingzeroECs}).  
We claim that there is a natural number $W$ such that $\max_s \Pr^{\max}_{s} (\Diamond wgt\geq W):=p<1$.  
Let $M:=\max_{s,\alpha} |wgt(s,\alpha)|$. Then, the claim follows as follows: 
For all $n\in\mathbb{N}$ we get that $\max_{s\in S} \Pr^{\max}_{\mathcal{M},s} (\Diamond wgt\geq n\cdot W+ M)\leq p^n$. 
Then the partial expectation of any scheduler is bounded by $\sum_{n=0}^\infty (n+1)\cdot W \cdot p^n =  \frac{W}{(1-p)^2}$.

For each end component $\mathcal{E}$, there is a number $W_\mathcal{E}$ and a probability $p_\mathcal{E}$ 
such that in $\mathcal{E}$ we have $\max_{s\in \mathcal{E}} \Pr^{\max}_{\mathcal{E},s} (\Diamond wgt\geq W_\mathcal{E}):=p_\mathcal{E}<1$. 
On the other hand, in the MEC-quotient of $\mathcal{M}$  the probability to reach $\mathit{goal}$ or $fail$ in $|S|$ steps is at least $\delta^{|S|}$ 
where $\delta$ is the minimal transition probability. 
Then we can conclude that 
\[\max_s \Pr^{\max}_{\mathit{MEC}(\mathcal{M}),s}(\Diamond wgt>M\cdot |S|)\leq 1-\delta ^{|S|}.\]
All in all, it is impossible for a scheduler to almost surely reach an accumulated weight above $ M \cdot |S|+\sum_{\mathcal{E}\text{ is an end component}} W_\mathcal{E}$.
\end{proof}

Recall that we define a \emph{critical scheduler} to be a scheduler $\mathfrak{S}$, for which there is a path containing a positive cycle, and for which $\Pr^{\mathfrak{S}}_{s_{\mathit{init}}}(\Diamond \mathit{goal})=0$

\begin{myprop}
Let $\mathcal{M}$ be an MDP. The optimal conditional expectation $\mathit{CE}^{\sup}_{s_{\mathit{init}}}=\infty$ if  $\mathit{PE}^{\sup}_{s_{\mathit{init}}}=\infty$ or if there is a critical scheduler $\mathfrak{S}$.
\end{myprop}

\begin{proof}
If $\mathit{PE}^{\sup}_{s_{\mathit{init}}}=\infty$ clearly also $\mathit{CE}^{\sup}_{s_{\mathit{init}}}=\infty$.
So, let $\mathfrak{S}$ be a scheduler which can reach a positive cycle but almost surely does not reach $\mathit{goal}$. 
Then, for any $n$, we can construct the following scheduler $\mathfrak{S}_n$. 
The scheduler $\mathfrak{S}_n$ attempts to reach the positive cycle directly, i.e. without visiting a state twice before. 
Then, it attempts to take the cycle $n$ times in a row. 
Only if  $\mathfrak{S}_n$ succeeds to do so, it maximizes the probability to reach the goal from then on. 
Otherwise, it avoids the $\mathit{goal}$. The scheduler $\mathfrak{S}_n$ reaches the goal with positive probability and  $\mathit{CE}^{\mathfrak{S}_n}_{s_{\mathit{init}}} \to \infty$ for $n\to \infty$.
\end{proof}

In Section \ref{sec:PE}, we gave the following construction:
Let $\mathcal{M}$ be an MDP with $\Pr^{\min}_{\mathcal{M},s_{\mathit{init}}}(\Diamond \mathit{goal})=0$ and $\mathit{CE}^{\sup}_{\mathcal{M},s_{\mathit{init}}}<\infty$. 
In particular, this means that there are no critical schedulers for $\mathcal{M}$. 
Let $S_0$ be the set of all states reachable from $s_{\mathit{init}}$ while only choosing actions in $Act^{\min}$. As there are no critical schedulers, $(S_0,Act^{\min})$ does not contain positive cycles. 
So, there is a unique maximal weight $w_s$ of paths leading from $s_{\mathit{init}}$ to  $s$ in  $S_0$.
Consider the following MDP $\mathcal{N}$: It contains the MDP $\mathcal{M}$ and a new initial state $t_{\mathit{init}}$.
For each $s\in S_0$ and each $\alpha\in Act(s)\setminus Act^{\min}(s)$, $\mathcal{N}$ also contains a new state $t_{s,\alpha}$ 
which is reachable from $t_{\mathit{init}}$ via an action $\beta_{s,\alpha}$ with weight $w_s$ and probability $1$. 
In $t_{s,\alpha}$, only action $\alpha$ with the same probability distribution over successors and the same weight as in $s$ is enabled. 
In this way, we ensure that $\Pr^{\min}_{\mathcal{N}, t_{\mathit{init}}} (\Diamond \mathit{goal}) >0$. 

\reductionPosPr*

\begin{proof} 
For each pair $(s,\alpha)$ with $s\in S_0$ and $\alpha\in Act(s)\setminus Act^{\min}(s)$, let $c_{s,\alpha}:=\sup_\mathfrak{S} \mathit{CE}^\mathfrak{S}_{\mathcal{N},t_{\mathit{init}}}$
 where the supremum is taken over all schedulers $\mathfrak{S}$ for $\mathcal{N}$ which assign probability $1$ to the action $\beta_{s,\alpha}$ in $t_{\mathit{init}}$.
Then, $\mathit{CE}^{\sup}_{\mathcal{N},t_{\mathit{init}}}=\max_{s,\alpha} c_{s,\alpha}=:c$.

A scheduler reaching the $\mathit{goal}$ with positive probability has to choose an action not in $Act^{\min}$ after at least one path.
Let $s\in S_0$ and $\alpha \in Act(t)\setminus Act^{\min}(s)$ be such that $c=c_{s,\alpha}$. For any scheduler $\mathfrak{T}$ for $\mathcal{N}$ 
starting with $\beta_{s,\alpha}$, we define the following scheduler $\mathfrak{T}^\prime$: 
$\mathfrak{T}^\prime$ starts by following a path with maximal accumulated weight from $s_{\mathit{init}}$ to $s$. 
If it reaches $s$ with accumulated weight $w_s$ it chooses $\alpha$ and follows the choices of $\mathfrak{T}$ from then on. 
If it does not reach $s$ with accumulated weight $w_s$, $\mathfrak{T}^\prime$ just picks actions in $Act^{\min}$, in this way making sure that the $\mathit{goal}$ will not be reached. 
In this way, $\mathit{CE}^{\mathfrak{T}^\prime}_{\mathcal{M},s_{\mathit{init}}}=\mathit{CE}^{\mathfrak{T}}_{\mathcal{N},t_{\mathit{init}}}$. So, $\mathit{CE}^{\sup}_{\mathcal{M},s_{\mathit{init}}} \geq c$.

Before we show the other direction, we define, given a finite path $\pi$, a finite path $\rho$ starting in $last(\pi)$, 
and a scheduler $\mathfrak{Q}$,  the scheduler $\mathfrak{Q}\uparrow\pi$ by
\[\mathfrak{Q}\uparrow\pi \,(\rho) := \mathfrak{Q}(\pi;\rho)\] where $\pi;\rho$ denotes the concatenation of the paths $\pi$ and $\rho$.

To show that for any scheduler $\mathfrak{S}$ for $\mathcal{M}$ with $\Pr^\mathfrak{S}_{s_{\mathit{init}}}>0$ 
we have $\mathit{CE}^\mathfrak{S}_{\mathcal{M},s_{\mathit{init}}}\leq c$, let $\mathfrak{S}$ be such a scheduler
 and consider the set $\Pi$ of finite $\mathfrak{S}$-paths $\pi$ in $\mathcal{M}^{\min}$ such that $\mathfrak{S}(\pi)\in Act(last(\pi))\setminus Act^{\min}$. 
 We know that for each $\pi\in\Pi$, 
 \[\frac{wgt(\pi)+ \mathit{PE}^{\mathfrak{S} \uparrow \pi}_{last(\pi)}}{\Pr^{\mathfrak{T}\uparrow \pi}_{last(\pi)}(\Diamond \mathit{goal})}\leq c.\] 
 We conclude that also
 \[
 \mathit{CE}^\mathfrak{S}_{s_{\mathit{init}}} = \frac{\sum_{\pi\in\Pi} \Pr^\mathfrak{S}_{s_{\mathit{init}}}(\pi)\cdot (wgt(\pi)+ \mathit{PE}^{\mathfrak{S} \uparrow \pi}_{last(\pi)})}{\sum_{\pi\in\Pi} \Pr^\mathfrak{S}_{s_{\mathit{init}}} (\pi) \cdot \Pr^{\mathfrak{S}\uparrow \pi}_{last(\pi)}(\Diamond \mathit{goal})}\leq c
 \]
 as all denominators are positive.
\end{proof}

\finiteCE*

\begin{proof}
We have seen that  $\mathit{CE}^{\sup}_{s_{\mathit{init}}}=\infty$  if there is a positively weight-divergent end component or a critical scheduler. 
On the other hand, we can rely on the reduction from the previous proposition if there are no critical scheduler. 
In $\mathcal{N}$, the maximal partial expectation is finite as there are no positively weight-divergent end components. 
As the minimal probability to reach $\mathit{goal}$ is furthermore positive, the maximal conditional expectation is finite as well. 
Hence, $\mathit{CE}^{\sup}_{\mathcal{M},s_{\mathit{init}}}=\mathit{CE}^{\sup}_{\mathcal{N},t_{\mathit{init}}}<\infty$.
\end{proof}

\subsection{Partial Expectations in MDPs with Non-Negative Weights} \label{app:non-negative}

Let $R\in\mathbb{Q}$ be arbitrary.
In this section, we  consider an MDP $\mathcal{M}$ in which all weights are non-negative, and we assume:
\begin{enumerate}
\item
$\mathit{PE}^{\sup}_{s_{init}} < \infty$,
\item
the only end components are the two distinct absorbing states $\mathit{goal}$ and $fail$,
\item
$\mathit{goal}$ can be reached from any state $s\in S\setminus \{fail\}$.
\end{enumerate}

Assumption 2. is justified as all weight are non-negative and hence the maximal expected mean payoff of an end component cannot be negative.

\tudparagraph{1ex.}{Saturation point.}
Recall that a saturation point for bias $R$ is a natural number $\mathfrak{p}$ such that there is a scheduler $\mathfrak{S}$ 
with  $\mathit{PE}^{\mathfrak{S}}_{s_{\mathit{init}}} [R] =\mathit{PE}^{{\sup}}_{s_{\mathit{init}}} [R] $ 
which is memoryless and deterministic as soon as the accumulated weight reaches $\mathfrak{p}$. 
I.e. for any two paths $\pi$ and $\pi^\prime$, with $last(\pi)=last(\pi^\prime)$ and $wgt(\pi),wgt(\pi^\prime)>\mathfrak{p}$, $\mathfrak{S}(\pi)=\mathfrak{S}(\pi^\prime)$. 
We first provide the following saturation point which we need in the proof of the smaller saturation point given in Section \ref{sec:PE}.

\begin{myprop} \label{prop:SaturationPoint1}
Let \[F:=\max_{s\in S} \mathit{PE}^{\sup}_s\] and \[E:= \mathrm{min}_{s\in S} \mathit{PE}^\mathfrak{Max}_s.\] 
Further, let \[\delta:=\mathrm{min}_{s\in S, \alpha\in Act(s)\setminus Act^{\max} (s)} p^{\max}_s - p^{\max}_{s,\alpha},\] where $p_{s,\alpha}^{\max}:= \sum_{t\in S} P(s,\alpha,t)\cdot p_t^{\max}$.
Then, \[\mathfrak{q}:= \frac{F-E}{\delta}-R\] is a saturation point for bias $R$.

(If the minimum in the definition of $\delta$ is taken over an empty set, $\mathfrak{Max}$ is already the optimal scheduler and hence any value is an upper saturation point.)
\end{myprop}

\begin{proof}
Given two schedulers $\mathfrak{S}$ and $\mathfrak{T}$ and some $x\in\mathbb{R}$, we define the scheduler $\mathfrak{S}\lhd_x\mathfrak{T}$ via:
\[\mathfrak{S}\lhd_x\mathfrak{T} (\pi) := 
\begin{cases}
\mathfrak{S}(\pi) & \text{if }wgt(\pi) \leq x, \\
\mathfrak{T}(\pi) & \text{otherwise.}
\end{cases}
\]
Given a finite path $\pi$ and a path $\rho$ starting in $last(\pi)$, we further define the scheduler $\mathfrak{S}\uparrow\pi$ by
\[\mathfrak{S}\uparrow\pi \,(\rho) := \mathfrak{S}(\pi;\rho)\] where $\pi;\rho$ denotes the concatenation of the paths $\pi$ and $\rho$.

Suppose, there is a scheduler $\mathfrak{S}$ such that $\mathit{PE}^{\mathfrak{S}}_{s_{\mathit{init}}} [R] >\mathit{PE}^{\mathfrak{S}\lhd_\mathfrak{q}\mathfrak{Max}}_{s_{\mathit{init}}} [R]$. 
By Proposition \ref{Scheduler Max}, this means that there must be a path $\pi$ with $wgt(\pi)\geq \mathfrak{q}$ such that $\mathfrak{S}(\pi)\not\in Act^{\max}(last(\pi))$ and
\[\mathit{PE}^{\mathfrak{S}\uparrow\pi}_{last(\pi)} [R]+ \Pr^{\mathfrak{S}\uparrow\pi}_{last(\pi)}(\Diamond \mathit{goal}) \cdot wgt(\pi) > \mathit{PE}^{\mathfrak{Max}}_{last(\pi)} [R] + p^{\max}_{last(\pi)}\cdot wgt(\pi).\]
But on the other hand, we have
\begin{eqnarray*}
&& \mathit{PE}^{\mathfrak{S}\uparrow\pi}_{last(\pi)} [R] + \Pr^{\mathfrak{S}\uparrow\pi}_{last(\pi)}(\Diamond \mathit{goal}) \cdot wgt(\pi) \\
&=& \mathit{PE}^{\mathfrak{S}\uparrow\pi}_{last(\pi)}+ \Pr^{\mathfrak{S}\uparrow\pi}_{last(\pi)}(\Diamond \mathit{goal}) \cdot (wgt(\pi) + R)\\
&\leq& F + (p^{\max}_{last(\pi)}-\delta) \cdot (wgt(\pi)+R)\\
&\leq& F+ p^{\max}_{last(\pi)}\cdot (wgt(\pi)+R) - \delta \cdot (\mathfrak{q} +R) \\
&=&  p^{\max}_{last(\pi)}\cdot (wgt(\pi)+R) + E \\
&\leq& \mathit{PE}^\mathfrak{Max}_{last(\pi)} + p^{\max}_{last(\pi)}\cdot (wgt(\pi)+R)\\
&=& \mathit{PE}^{\mathfrak{Max}}_{last(\pi)} [R] + p^{\max}_{last(\pi)}\cdot wgt(\pi).
\end{eqnarray*}
contradicting the supposition. For the third line, note that $wgt(\pi)+R > 0 $.
\end{proof}

In the setting of non-negative weights, it has been shown in \cite{chen2013} that there is a weight-based deterministic scheduler maximizing the partial expectation. 
Of course, this also follows from our results in Section \ref{sec:optimalWD} for MDPs with arbitrary weights. This allows us to conclude the following.
\begin{mycor}\label{FinitelyManySchedulers}
The supremum in $\mathit{PE}^{{\max},R}_s:=\mathrm{sup}_\mathfrak{S} \mathit{PE}^{\mathfrak{S},R}_s$ can also be taken over weight-based schedulers 
which behave memoryless as soon as the accumulated weight reaches $\mathfrak{q}$. 
As there are only finitely many such schedulers, the supremum is furthermore in fact a maximum.
\end{mycor}

\begin{myprop} \label{SaturationPoint}
Let $\mathcal{M}$ be an MDP with non-negative weights satisfying the assumptions 1--3. Furthermore, let 
\[p_{s,\alpha}^{\max}:= \sum_{t\in S} \Pr(s,\alpha,t)\cdot p_t^{\max},\] and \[\mathit{PE}^\mathfrak{Max}_{s,\alpha}:=p_{s,\alpha}^{\max}\cdot wgt(s,\alpha) + \sum_{t\in S} \Pr(s,\alpha,t)\cdot  \mathit{PE}^\mathfrak{Max}_t.\]
Then, \[\mathfrak{p}_R:= \mathrm{sup} \left\{\left.\frac{\mathit{PE}^\mathfrak{Max}_{s,\alpha}-E^\mathfrak{Max}_{s}}{p^{\max}_s-p^{\max}_{s,\alpha}}\right|s\in S, \alpha\in Act(s)\setminus Act^{\max}(s)\right\}-R\] 
is an upper saturation point for $\mathcal{M}$.
\end{myprop}

\begin{proof}
It is enough to show that
\begin{equation} \label{eq:ProofSaturationPoint}
\mathbb{E}^\mathfrak{Max}_s(\oplus^{\mathfrak{p}_R+R} \mathit{goal}) \geq \mathbb{E}^\mathfrak{S}_s(\oplus^{\mathfrak{p}_R+R} \mathit{goal})
\end{equation}
 for all states $s$ and all schedulers $\mathfrak{S}$. Since $\mathfrak{Max}$ maximizes the probability of reaching the goal, 
 this  implies that $\mathbb{E}^\mathfrak{Max}_s(\oplus^{w+R} \mathit{goal}) \geq \mathbb{E}^\mathfrak{S}_s(\oplus^{w+R} \mathit{goal})$ for all $w\geq \mathfrak{p}_R$.

So, we show (\ref{eq:ProofSaturationPoint}). Suppose that
\[\mathbb{E}^\mathfrak{Max}_s(\oplus^{\mathfrak{p}_R+R} \mathit{goal}) < \mathbb{E}^\mathfrak{S}_s(\oplus^{\mathfrak{p}_R+R} \mathit{goal})\] 
for some state $s$ and some scheduler $\mathfrak{S}$. 

Now, let  \[D_s:=\mathrm{sup}_\mathfrak{T} \mathbb{E}^\mathfrak{T}_s (\oplus^{\mathfrak{p}_R+R} \mathit{goal}) - \mathbb{E}^\mathfrak{Max}_s (\oplus^{\mathfrak{p}_R+R} \mathit{goal})\] for each state $s$. 
As we know that there are only finitely many relevant schedulers for the supremum we can actually choose a (deterministic weight-based) scheduler $\mathfrak{S}$ 
such that $\mathrm{sup}_\mathfrak{T} \mathbb{E}^\mathfrak{T}_s  (\oplus^{\mathfrak{p}_R+R} \mathit{goal})=  \mathbb{E}^\mathfrak{S}_s (\oplus^{\mathfrak{p}_R+R} \mathit{goal})$ for all states $s$. 
Now, let $t$ be a state such that $D_t> 0$ is maximal, and such that the first action  $\alpha$ that $\mathfrak{S}$ chooses starting in $t$ leads to a state $r$ with $D_r < D_t$ with positive probability. 
As $D_{\mathit{goal}}= 0$ such a state exists. Then,
\begin{eqnarray*}
&& \mathbb{E}^\mathfrak{S}_t (\oplus^{\mathfrak{p}_R+R} \mathit{goal})  \\
&=& \Pr^\mathfrak{S}_t(\Diamond \mathit{goal}) \cdot wgt(t,\alpha) + \sum_s P(t,\alpha,s) \cdot \mathbb{E}^\mathfrak{S}_s (\oplus^{\mathfrak{p}_R+R} \mathit{goal}) \\
&\leq& p^{\max}_{t,\alpha} \cdot wgt(t,\alpha) + \sum_s P(t,\alpha,s) \cdot (\mathbb{E}^\mathfrak{Max}_s (\oplus^{\mathfrak{p}_R+R} \mathit{goal}) + D_s) \\
&<& p^{\max}_{t,\alpha} \cdot wgt(t,\alpha) + \left(\sum_s P(t,\alpha,s) \cdot \mathbb{E}^\mathfrak{Max}_s (\oplus^{\mathfrak{p}_R+R} \mathit{goal})\right) + D_t \\
&=& \mathit{PE}^\mathfrak{Max}_{t,\alpha} +  p^{\max}_{t,\alpha} \cdot (\mathfrak {p}_R+R) +D_t \\
&=& \mathit{PE}^\mathfrak{Max}_{t,\alpha} +  p^{\max}_{t,\alpha} \cdot (\mathfrak {p}_R+R) + \mathbb{E}^\mathfrak{S}_t (\oplus^{\mathfrak{p}_R+R} \mathit{goal}) - \mathbb{E}^\mathfrak{Max}_t (\oplus^{\mathfrak{p}_R+R} \mathit{goal})\\
&=& \mathit{PE}^\mathfrak{Max}_{t,\alpha} +  p^{\max}_{t,\alpha} \cdot (\mathfrak {p}_R+R) + \mathbb{E}^\mathfrak{S}_t (\oplus^{\mathfrak{p}_R+R} \mathit{goal}) - (E^\mathfrak{Max}_t + p^{\max}_t \cdot (\mathfrak{p}_R+R)).
\end{eqnarray*}

Subtracting $\mathbb{E}^\mathfrak{S}_t (\oplus^{\mathfrak{p}_R+R} \mathit{goal})$, we get 
\begin{eqnarray*}
0 &<& \mathit{PE}^\mathfrak{Max}_{t,\alpha} +  p^{\max}_{t,\alpha} \cdot (\mathfrak{p}_R+R) - (\mathit{PE}^\mathfrak{Max}_t + p^{\max}_t \cdot (\mathfrak{p}_R+R)) \\ 
 &=& (\mathit{PE}^\mathfrak{Max}_{t,\alpha}- \mathit{PE}^\mathfrak{Max}_t) - (p^{\max}_t -  p^{\max}_{t,\alpha}) \cdot (\mathfrak{p}_R+R) .
\end{eqnarray*}
But, the right hand side evaluates to $0$ by the definition of $\mathfrak{p}_R$ leading to a contradiction.
\end{proof}

\tudparagraph{1ex.}{Computation of the Partial Expectation.} 

\begin{myprop} \label{prop:linearprogram}
Let $R\in\mathbb{Q}$ and let $B_R$ be the least integer greater or equal to $\mathfrak{p}_R+ \max_{s\in S,\alpha\in Act(s)} wgt(s,\alpha)$ and let $S^\prime := S\setminus\{\mathit{goal},fail\}$. 
Consider the following linear program in the variables $(x_{s,r})_{s\in S^\prime, 0\leq r\leq B_R}$ (r ranges over integers): 
Minimize $\sum_{s\in S^\prime,0\leq r\leq B_R}x_{s,r}$ under the following constraints:
\begin{eqnarray*}
&&\text{For }r \geq \mathfrak{p}_R:  x_{s,r}= p_s^{\max} \cdot (r+R) + E^\mathfrak{Max}_s, \\
&&\text{for }r < \mathfrak{p}_R \text{ and }\alpha\in Act(s) :\\
&&x_{s,r}\geq P(s,\alpha,\mathit{goal}) \cdot (r+R+wgt(s,\alpha)) + \sum_{t\in S^\prime} P(s,\alpha,t) \cdot x_{t,r+wgt(s,\alpha).}
\end{eqnarray*}

The values $(\mathit{PE}^{\sup}_{s_{\mathit{init}}}[r+R])_{s\in S, 0\leq r\leq B_R}$ form the unique solution to this linear program. 
\end{myprop} 

We prove the unique solvability of this linear program in detail. Linear programs claimed to be uniquely solvable below can be treated analogously.

\begin{proof} 
Following a standard approach  by  Veinott \cite{Veinott1969}, 
we want to show that the linear program is uniquely solvable by defining a contraction mapping with respect to a weighted maximum norm whose fixed point is the optimal solution. 
The definition of the weights we use is made explicit by Tseng \cite{tseng1990}.

We define a function $T^R:\mathbb{R}^{S^\prime\times\{0,\dots,B_R\}} \to \mathbb{R}^{S^\prime\times\{0,\dots,B_R\}}$. For $s\in S$ and $r\leq B_R$, let 
\[ T^R_{s,r}(x) := \begin{cases}
\max\limits_{\alpha\in Act(s)} \Big( \sum_{t\in S^\prime} P(s,\alpha,t) \cdot x_{t,r+wgt(s,\alpha)} & \\
\hspace{1.5cm} + P(s,\alpha,\mathit{goal}) \cdot (r+R+wgt(s,\alpha)) \Big) & \text{if } r<\mathfrak{p}_R,\\
p_s^{\max} \cdot (r+R) + E^\mathfrak{Max}_s & \text{otherwise.}
\end{cases}\]

In order to define a suitable weighted maximum norm, we begin by recursively defining the following partition $S_0,\dots , S_k$ of $S$:
\begin{eqnarray*}
S_0 &:=& \{\mathit{goal}, fail\}\\
S_{i+1} &:=& \{t\in S\setminus (S_0 \cup \dots \cup S_i) | \sum_{q\in S_0 \cup \dots \cup S_i} P(t,\alpha,q) >0 \text{ for all }\alpha \in Act(t)\}. 
\end{eqnarray*}
If at some point $S_0,\dots,S_i$ is not yet a partition of $S$, then $S_{i+1}$ is non-empty. If it was empty, then $T:=S\setminus(S_0 \cup \dots \cup S_i)$ 
would contain an end component, as for each $t\in T$ there would be an action $\alpha$ such that $P(t,\alpha,T)=1$. 
So, the recursive definition produces a partition $S=S_0\cup \dots \cup S_k$ in finitely many steps.

Now, for each $s\in S_i$ let $w_s:= 1 - \delta^{2i}$ where \[\delta := \mathrm{min} \{P(s,\alpha,t) | s, \alpha, t \text{ s.t. } P(s,\alpha,t)>0\}  .\] 
These $w_s$ will serve as weights for our weighted supremum norm.

We use the following fact \cite[Lemma 3]{tseng1990}:
Let $\gamma:= \frac{1-\delta^{2k-1}}{1-\delta^{2k}}<1$.
For all $s\not\in\{gaol, fail\}$ and all $\alpha\in Act(s)$, we have 
\[ \sum_{t\in S^\prime} \frac{P(s,\alpha,t)\cdot w_t}{w_s}\leq \gamma.
\]

We define the following norm on $\mathbb{R}^{S^\prime\times\{0,\dots,B_R\}}$:
\[\|x\|_w := \max\limits_{s\in S^\prime,0\leq r\leq B_R} \frac{|x_{s,r}|}{w_s}.\]

We show that $T^R$ is a contraction with respect to this norm: For $x,y\in \mathbb{R}^{S^\prime\times\{0,\dots,B_R\}}$, we claim
\[ \|T^R(x) - T^R(y)\|_w \leq \gamma\cdot\|x-y\|_w. \]

Let $s\in S^\prime$ and $r<\mathfrak{p}_R$.
\begin{eqnarray*}
& & T^R_{s,r}(x)- T^R_{s,r}(y) \\
& = &\max\limits_{\alpha\in Act(s)} \Big( \sum_{t\in S^\prime} P(s,\alpha,t) \cdot x_{t,r+wgt(s,\alpha)}+ P(s,\alpha,\mathit{goal}) \cdot (r+R+wgt(s,\alpha)) \Big)\\
 & &-\max_{\beta\in Act(s)}  \Big(\sum_{q\in S^\prime} P(s,\beta,q) \cdot y_{q,r+wgt(s,\beta)}+ P(s,\beta,\mathit{goal} )\cdot (r+R+wgt(s,\beta)) \Big)\\
& \leq &  \max_{\alpha\in Act(s)} \sum_{t\in S^\prime} P(s,\alpha,t) \cdot (x_{t,r+wgt(s,\alpha)} -  y_{t,r+wgt(s,\alpha)})\\
& = &  \max_{\alpha\in Act(s)} \sum_{t\in S^\prime} (P(s,\alpha,t)\cdot w_t) \cdot \frac{(x_{t,r+wgt(s,\alpha)} -  y_{t,r+wgt(s,\alpha)})}{w_t}\\
& \leq &  \max_{\alpha\in Act(s)} \sum_{t\in S^\prime} (P(s,\alpha,t)\cdot w_t) \cdot \max_{q\in S}\frac{ (x_{q,r+wgt(s,\alpha)} -  y_{q,r+wgt(s,\alpha)})}{w_q}\\
& \leq & \gamma \cdot w_s \cdot \|x - y\|_w.
\end{eqnarray*}
By symmetry, we obtain the same inequality for $T^R_{s,r}(y)- T^R_{s,r}(x) $. Dividing the inequality by $w_s$ concludes the proof that $T^R$ is a contraction. 
By the Banach fixed-point theorem, $T^R$ hence has a unique fixed point.

That this fixed point is the unique solution of the linear program is now easy to see. The map $T^R$ is defined such that any  $z$ satisfying the constraints of the linear program satisfies $z\geq T^R(z)$. 
But, if there are coordinates $s,r$ such that $z_{s,r}>T^R_{s,r}(z)$ then replacing $z_{s,r}$ by $T^R_{s,r}(z)$ leads to $z^\prime$ still satisfying the constraints and resulting in a smaller value of $\sum_{s,r} z^\prime_{s,r}$. 
So, the unique fixed point of $T^R$ is the unique optimal solution of the linear program.

Finally, we can easily check that $(\mathbb{E}^{\max}_s(\oplus^{r+R} \mathit{goal}))_{s\in S, 0\leq r\leq B_R}$ is indeed a fixed point of $T^R$ as we know that $\mathfrak{p}_R$ is a saturation point.
\end{proof}


\section{Existence of Optimal Schedulers} \label{app:optimalWD}


We provide the proofs to Section \ref{sec:optimalWD} here.

Recall that we consider an MDP $\mathcal{M}$ with finite maximal partial expectation. 
In particular, we assume that all states are reachable from $s_{\mathit{init}}$ and that $\mathit{goal}$ is reachable from all states except $fail$. 
Furthermore, there are no positively weight-divergent end components and so we can assume that all end components have negative maximal expected mean payoff (see Proposition \ref{prop:removingzeroECs}).

We split the proof of Proposition \ref{prop:supWD} into the following two propositions:

\begin{restatable}{myprop} {WRscheduler}
Let $\mathcal{M}$ be an MDP with $\mathit{PE}^{\sup}_{s_{\mathit{init}}}<\infty$.
For each scheduler $\mathfrak{S}\in HR^\mathcal{M}$, there is a scheduler $\mathfrak{T}\in WR^\mathcal{M}$ such that $\mathit{PE}^\mathfrak{S}=\mathit{PE}^\mathfrak{T}$ and $\Pr^\mathfrak{S}_{s_{\mathit{init}}}(\Diamond \mathit{goal})=\Pr^\mathfrak{T}_{\sinit} (\Diamond \mathit{goal})$.
\end{restatable}

\begin{proof}
Let $\mathfrak{S}\in HR$. For each state-weight pair $(s,w)$ with $s\in S\setminus \{\mathit{goal}, fail\}$ and $w\in\mathbb{Z}$, we let $\theta^\mathfrak{S}_{s,w}$ 
be the expected number of times that $s$ is reached with accumulated weight $w$ under $\mathfrak{S}$, and we let $\theta^\mathfrak{S}_{s,w,\alpha}$ be the expected number of times that $\alpha$ is chosen in this situation by $\mathfrak{S}$. 
We have that
\[ \theta^\mathfrak{S}_{s,w} = \sum\limits_{\substack{\pi \text{ finite path,} \\ last(\pi)= s, \\ wgt(\pi)=w }} \Pr^\mathfrak{S}_{s_{\mathit{init}}} (\pi) 
\text{ and } 
\theta^\mathfrak{S}_{s,w,\alpha} = \sum\limits_{\substack{\pi \text{ finite path,} \\ last(\pi)= s, \\ wgt(\pi)=w }} \Pr^\mathfrak{S}_{s_{\mathit{init}}} (\pi) \cdot \mathfrak{S}(\pi)(\alpha).\]
Note that $\theta^\mathfrak{S}_{s,w}$ is finite for all $s\in S\setminus \{\mathit{goal}, fail\}$, $w\in \mathbb{Z}$ as all end components have negative maximal expected mean payoff. 

Now, we define a weight-based deterministic scheduler $\mathfrak{T}$ by 
\[ \mathfrak{T}(s,w)(\alpha) := \begin{cases}
\theta^\mathfrak{S}_{s,w,\alpha} / \theta^\mathfrak{S}_{s,w} & \text{, if }\theta^\mathfrak{S}_{s,w}>0, \\
\text{arbitrary }& \text{, otherwise.} 
\end{cases}\]
Clearly, only state-weight-pairs $(s,w)$ which are reachable under $\mathfrak{S}$ are reachable under $\mathfrak{T}$. Further, $\mathfrak{T}$ is well-defined as $\sum_{\alpha\in Act(s)} \mathfrak{T}(s,w)(\alpha) =1$ for all reachable $(s,w)$. 

For each state-weight-pair $(s,w)$ reachable under $\mathfrak{S}$ and hence under $\mathfrak{T}$ let $\theta^{\mathfrak{T}}_{s,w}$ be the expected number of times that $(s,w)$ is reached under $\mathfrak{T}$.  
Then, the collection of all $\theta^{\mathfrak{T}}_{s,w}$ is uniquely determined by the following set of equations: For all $s,w$,
\[ \theta^{\mathfrak{T}}_{s,w} = \delta_{s,w} + \sum_{t\in S, \alpha\in Act(t)} P(t,\alpha,s)\cdot x_{t,w- wgt(t,\alpha)}\cdot \mathfrak{T}(t,w-wgt(t,\alpha))(\alpha)\]
where $\delta_{s,w}=1$ iff $s=s_{\mathit{init}}$ and $w=0$, and  $\delta_{s,w}=0$ otherwise.
By spelling out the last steps of the paths in the definition of $\theta_{s,w}^\mathfrak{S}$, one can see that $\theta_{s,w}^\mathfrak{S}$ provides the solution to this set of equations and hence  $\theta_{s,w}^\mathfrak{S}= \theta_{s,w}^\mathfrak{T}$ for all $(s,w)$. 
By the definition of $\mathfrak{T}$, the expected number of times action $\alpha$ is chosen in $(s,w)$ under $\mathfrak{T}$ is hence  $\theta_{s,w,\alpha}^\mathfrak{S}$ as well and the claim follows.
\end{proof}

\begin{restatable} {myprop}{WDscheduler}
Let $\mathcal{M}$ be an MDP with $\mathit{PE}^{\sup}_{s_{\mathit{init}}}<\infty$. Then, we have
\[\sup_{\mathfrak{S}\in \mathit{WR}^{\mathcal{M}}} \mathit{PE}^\mathfrak{S}_{s_{\mathit{init}}} = \sup_{\mathfrak{S}\in \mathit{WD}^{\mathcal{M}}} \mathit{PE}^\mathfrak{S}_{s_{\mathit{init}}}.\]
\end{restatable}

\begin{proof}
Let $\mathfrak{S}$ be a weight-based randomized scheduler and let $(s,w)$ be a state-weight-pair  such that $\mathfrak{S}(s,w)$ is not a Dirac distribution.
We define $\mathfrak{S}\uparrow w$ by
\[\mathfrak{S}\uparrow w \,(t,v) := \mathfrak{S}(t,v+w).\] 
Now, 
\begin{eqnarray*}
&&\mathit{PE}^{\mathfrak{S}\uparrow w}_{s} [w]  \\
=&& \sum_{\alpha\in Act(s)} \mathfrak{S}(s,w)(\alpha) \cdot \sum_{t\in S} P(s,\alpha,t) \cdot \mathit{PE}^{\mathfrak{S}\uparrow(w+wgt(s,\alpha))}_{t} [w+wgt(s,\alpha)].
\end{eqnarray*}
But then there is an action $\beta\in Act(last(\pi))$ such that 
\begin{eqnarray*}
&&\sum_{t\in S} P(s,\beta,t) \cdot \mathit{PE}^{\mathfrak{S}\uparrow(w+wgt(s,\beta))}_{t} [w+wgt(s,\beta)] \\
\geq &&\sum_{\alpha\in Act(s)} \mathfrak{S}(s,w)(\alpha) \cdot \sum_{t\in S} P(s,\alpha,t) \cdot \mathit{PE}^{\mathfrak{S}\uparrow(w+wgt(s,\alpha))}_{t} [w+wgt(s,\alpha)]
\end{eqnarray*}
because $\mathfrak{S}(s,w)$ is a probability distribution. 
We conclude that the scheduler $\mathfrak{S}^\prime$ which agrees with $\mathfrak{S}$ on all state-weight-pairs except $(s,w)$ and assigns probability $1$ to $\beta$ for $(s,w)$ satisfies $\mathit{PE}^{\mathfrak{S}^\prime}_{s_{\mathit{init}}} \geq \mathit{PE}^{\mathfrak{S}}_{s_{\mathit{init}}}$. In this way, we can replace all probability distributions that $\mathfrak{S}$ chooses by Dirac distributions and generate a sequence of schedulers with non-decreasing partial expectations. Ultimately, we obtain a weight-based deterministic scheduler $\mathfrak{T}$ with $\mathit{PE}^{\mathfrak{T}}_{s_{\mathit{init}}}\geq\mathit{PE}^{\mathfrak{S}}_{s_{\mathit{init}}}$.
\end{proof}

\begin{mydef}[Metric on weight-based deterministic schedulers]
Given an MDP $\mathcal{M}$ with arbitrary integer weights, we define the following metric $d^\mathcal{M}$ on the set of weight-based deterministic schedulers, i.e. on the set of functions from $S\times \mathbb{Z} \to Act$:
For two such schedulers $\mathfrak{S}$ and $\mathfrak{T}$, we let \[ d^\mathcal{M}(\mathfrak{S},\mathfrak{T}) := 2^{-R} \] where $R$ is the greatest natural number such that \[\mathfrak{S}\upharpoonright S\times \{-(R-1), \dots , R-1\} = \mathfrak{T}\upharpoonright S\times \{-(R-1), \dots , R-1\}\] or $\infty$ if there is no greatest such natural number.
\end{mydef}

\begin{mylem}
The metric space $(Act^{S\times \mathbb{Z}},d^\mathcal{M})$ is compact.
\end{mylem}
\begin{proof}
We can identify $Act^{S\times \mathbb{Z}}$ with $(Act^{S\times \{+,-\}})^\mathbb{N}$. Then it is easy to see that the metric $d^\mathcal{M}$ induces the usual tree topology on this finitely branching tree of height $\omega$. Therefore, the space is homeomorphic to the Cantor space $2^\omega$ and hence compact.
\end{proof}

\begin{mylem}[Upper Semi-Continuity of Partial Expectations]
Assume that $\mathit{PE}^{\sup}_{s_{\mathit{init}}}$ is finite in $\mathcal{M}$. Then, the function \[\mathit{PE}: (\mathit{WD}, d^{\mathit{WD}}) \to (\mathbb{R}_\infty,d^{euclid})\] assigning $\mathit{PE}^\mathfrak{S}_{s_{\mathit{init}}}$ to a weight-based deterministic scheduler $\mathfrak{S}$ is upper semi-continuous. 

\end{mylem}

\begin{proof}
Let $\mathfrak{S}$ be a $WD$-scheduler with $\mathit{PE}^\mathfrak{S}_{s_{\mathit{init}}}>-\infty$. Given $\epsilon>0$, we will define a natural number $R$ such that any $WD$-scheduler $\mathfrak{T}$ with $\mathfrak{T}\upharpoonright S\times [-R,R] = \mathfrak{S}\upharpoonright S\times [-R,R]$ satisfies $\mathit{PE}^\mathfrak{T}_{s_{\mathit{init}}} < \mathit{PE}^\mathfrak{S}_{s_{\mathit{init}}} + 4 \epsilon$.

First, we observe that for each state $s$ there is a natural number $W_s$ and a probability $p_s<1$ such that $\Pr^{\max}_s (\Diamond wgt>W_s)\leq p_s$ because $\mathit{PE}^{\sup}_s<\infty$. Let $W:= \max_s W_s + \max_{s,\alpha} |wgt(s,\alpha)|$ and $p:= \max_s p_s$. Then, for each state $s$ and each natural number $n$, we have that $\Pr^{\max}_s(\Diamond wgt > n\cdot W)\leq p^n$.

As $\mathit{PE}^\mathfrak{S}_{s_{\mathit{init}}}>-\infty$, we know that 
\[\lim_{l\to\infty} \sum_{n=l}^\infty \Pr^\mathfrak{S}_{s_{\mathit{init}}} (\Diamond^{=-n} \mathit{goal})\cdot n =0.
\]
We define $l_\epsilon^\mathfrak{S}$ to be the smallest natural number such that 
\[\sum_{n=l_\epsilon^\mathfrak{S}}^\infty \Pr^\mathfrak{S}_{s_{\mathit{init}}} (\Diamond^{=-n} \mathit{goal})\cdot n <\epsilon.\]
Let $k$ be the smallest natural number such that $p^k\cdot(l^\mathfrak{S}_\epsilon+W/(1-p)^2) < \epsilon$. Further define $R^- :=l^\mathfrak{S}_\epsilon+k\cdot W$.
Let 
\begin{eqnarray*}
H:= \min \{0, \mathit{PE}^{\mathfrak{S}\uparrow r}_s(\oplus^r \mathit{goal}) | && s\in S, 0\leq r\leq \max_{s,\alpha} |wgt(s,\alpha)|, \\
&& \text{ and }(s,r)\text{ is reachable under }\mathfrak{S}\}, 
\end{eqnarray*}
and let $n$ be the least natural number such that 
\[ p^n \cdot (n\cdot W \cdot \max_s p_s^{\max} + \max_s \mathit{PE}^{\sup}_s -H) < 2\epsilon.\]
Finally, define $R^+:= n\cdot W$.

We claim that $R:=\max\{R^-,R^+\}$ does the job. So let $\mathfrak{T}$ be a scheduler with $\mathfrak{T}\upharpoonright S\times [-R,R] = \mathfrak{S}\upharpoonright S\times [-R,R]$.
Let 
\begin{eqnarray*}
\mathfrak{P}^+ := \{\pi \text{ finite path} | && wgt(\pi)>R\text{ and any proper prefix }\pi^\prime \\
&& \text{ satisfies }wgt(\pi^\prime)\in[-R,R]\}, \\
\mathfrak{P}^- := \{\pi \text{ finite path}| && wgt(\pi)<R\text{ and any proper prefix }\pi^\prime \\
&& \text{ satisfies }wgt(\pi^\prime)\in[-R,R]\}.
\end{eqnarray*}

Recall the following definition. Given a finite path $\pi$ and a path $\rho$ starting in  $last(\pi)$ and a scheduler $\mathfrak{Q}$, we further define the scheduler $\mathfrak{Q}\uparrow\pi$ by
\[\mathfrak{Q}\uparrow\pi \,(\rho) := \mathfrak{Q}(\pi;\rho)\] where $\pi;\rho$ denotes the concatenation of the paths $\pi$ and $\rho$.

The schedulers $\mathfrak{S}$ and $\mathfrak{T}$ agree on all paths not having a prefix in one of these two sets. 
So,
\begin{eqnarray*}
&&\mathit{PE}^\mathfrak{T}_{s_{\mathit{init}}}-\mathit{PE}^\mathfrak{S}_{s_{\mathit{init}}} \\
=&&\sum_{\pi\in\mathfrak{P}^+\cup\mathfrak{P}^-}  (\mathbb{E}^\mathfrak{T\uparrow \pi}_{last(\pi)}(\oplus^{wgt(\pi)}\mathit{goal})-\mathbb{E}^\mathfrak{S\uparrow \pi}_{last(\pi)}(\oplus^{wgt(\pi)}\mathit{goal}))\cdot \Pr^\mathfrak{S}_{s_{\mathit{init}}}(\pi)
\end{eqnarray*}
Split up into two sums, we get the following inequalities:
\begin{eqnarray*}
&&\sum_{\pi\in\mathfrak{P}^+}  (\mathbb{E}^\mathfrak{T\uparrow \pi}_{last(\pi)}(\oplus^{wgt(\pi)}\mathit{goal})-\mathbb{E}^\mathfrak{S\uparrow \pi}_{last(\pi)}(\oplus^{wgt(\pi)}\mathit{goal}))\cdot \Pr^\mathfrak{S}_{s_{\mathit{init}}}(\pi)\\
&\leq&\Pr^\mathfrak{S}_{s_{\mathit{init}}}(\mathfrak{P}^+) \cdot (R \cdot \max_s p_s^{\max}+\max_s \mathit{PE}^{\sup}_s - \min_{\pi\in\mathfrak{P}^+}\mathbb{E}^\mathfrak{S\uparrow \pi}_{last(\pi)}(\oplus^{wgt(\pi)}\mathit{goal}))\\
&\leq&p^n\cdot (n\cdot W \cdot \max_s p_s^{\max}+\max_s \mathit{PE}^{\sup}_s - H) < 2\epsilon.
\end{eqnarray*}

\begin{eqnarray*}
&&\sum_{\pi\in\mathfrak{P}^-}  (\mathbb{E}^\mathfrak{T\uparrow \pi}_{last(\pi)}(\oplus^{wgt(\pi)}\mathit{goal})-\mathbb{E}^\mathfrak{S\uparrow \pi}_{last(\pi)}(\oplus^{wgt(\pi)}\mathit{goal}))\cdot \Pr^\mathfrak{S}_{s_{\mathit{init}}}(\pi)\\
&=&\sum_{\pi\in\mathfrak{P}^-}  \mathbb{E}^\mathfrak{T\uparrow \pi}_{last(\pi)}(\oplus^{wgt(\pi)}\mathit{goal})\cdot \Pr^\mathfrak{S}_{s_{\mathit{init}}}(\pi) \\
&&-\sum_{\pi\in\mathfrak{P}^-}  \mathbb{E}^\mathfrak{S\uparrow \pi}_{last(\pi)}(\oplus^{wgt(\pi)}\mathit{goal})\cdot \Pr^\mathfrak{S}_{s_{\mathit{init}}}(\pi)\\
&\leq& \sum_{i=0}^\infty \max_{\pi\in\mathfrak{P}} \Pr^{\mathfrak{T}\uparrow \pi}_{last(\pi)}(\Diamond^{=R^-+i} \mathit{goal})\cdot i \\
&& - \sum_{\pi\in\mathfrak{P}^-}  \mathbb{E}^\mathfrak{S\uparrow \pi}_{last(\pi)}(\oplus^{wgt(\pi)}\mathit{goal})\cdot \Pr^\mathfrak{S}_{s_{\mathit{init}}}(\pi)\\
&\leq& \sum_{i=0}^\infty p^{k+i} \cdot W \cdot (i+1) - \sum_{n=l_\epsilon^\mathfrak{S}+1}^\infty \Pr^\mathfrak{S}_{s_{\mathit{init}}}(\{\pi\in\mathfrak{P}^-|\pi\vDash\Diamond^{-n} \mathit{goal}\})\cdot (-n) \\
&&- \sum_{n=-l_\epsilon^\mathfrak{S}}^\infty \Pr^\mathfrak{S}_{s_{\mathit{init}}}(\{\pi\in\mathfrak{P}^-|\pi\vDash\Diamond^{n} \mathit{goal}\})\cdot n \\
&\leq& p^k \cdot W/(1-p)^2 + \epsilon + p^k \cdot l^\mathfrak{S}_\epsilon \leq 2\epsilon.
\end{eqnarray*}
This finishes the case where $\mathit{PE}^\mathfrak{S}_{s_{\mathit{init}}}>-\infty$.
If $\mathit{PE}^\mathfrak{S}_{s_{\mathit{init}}}=-\infty$, we have to show for each $b\in\mathbb{R}$ that there is an $R$ such that all schedulers agreeing with $\mathfrak{S}$ on the weight-window $[-R,R]$ have a partial expectation below $b$. But as we can make \[\sum_{\zeta\vDash\Diamond \mathit{goal} \land \Box wgt \in [-R,R]} wgt(\zeta)\cdot \Pr^\mathfrak{S}_{s_{\mathit{init}}}(\zeta)\] arbitrarily small, this is easy to see.
\end{proof}

\begin{mythm}[Existence of Optimal Schedulers for Partial Expectations]
If $\mathit{PE}^{\sup}_{s_{\mathit{init}}}$ is finite in an MDP $\mathcal{M}$, then there is a weight-based deterministic scheduler $\mathfrak{S}$ with 
\[\mathit{PE}^{\sup}_{s_{\mathit{init}}}=\mathit{PE}^\mathfrak{S}_{s_{\mathit{init}}}.\]
\end{mythm}

\begin{proof}
If $\mathit{PE}^{\sup}_{s_{\mathit{init}}}$ is finite, then the map $\mathit{PE}: (WD, d^{WD}) \to (\mathbb{R}_\infty,d^{euclid})$ is upper semi-continuous as we have just shown. But then, this map has a maximum because $(WD,d^{WD})$ is a compact metric space. 
\end{proof}


\section{Approximation}  \label{app:approximation}


\subsection{Approximating Optimal Partial Expectations} \label{app:approximationPE}


\begin{myprop}
Let $s\in S$ and let \[\mathfrak{q}_s:=\frac{\mathit{PE}^{ub}-\mathit{PE}^{\mathfrak{Min}}_s}{p_s^{\min} - \min\limits_{\alpha\not\in Act^{\min}(s)} p_{s,\alpha}^{\min}}.\] Then any $WD$-scheduler $\mathfrak{S}$ maximizing the partial expectation in $\mathcal{M}$ satisfies $\mathfrak{S}(s,w)\in Act^{\min}(s)$ if $w\leq \mathfrak{q}_s$.
\end{myprop}

\begin{proof}
The proof works  analogously to the proof of Proposition \ref{prop:SaturationPoint1}.
\end{proof}

Let $\mathfrak{q}:=\min_{s\in S} \mathfrak{q}_s$ and let $D:=\mathit{PE}^{ub}-\min \{ \mathit{PE}^\mathfrak{Max}_s , \mathit{PE}^\mathfrak{Min}_s |s\in S\}$.
Given $\epsilon>0$, we define $R^+_\epsilon := (c_\mathcal{M}+W)\cdot \left\lceil \frac{\log(2D)+\log(1/\epsilon)}{\log(1/\lambda_\mathcal{M})} \right\rceil$ and $R^-_\epsilon:= \mathfrak{q}-R^+_\epsilon$.

\begin{mythm}
There is a weight-based deterministic scheduler $\mathfrak{S}$ such that the scheduler $\mathfrak{T}$ defined by
\[\mathfrak{T} (\pi) = 
\begin{cases}
\mathfrak{S}(\pi) & \text{, if any prefix }\pi^\prime\text{ of }\pi\text{ satisfies }R^-_\epsilon\leq wgt(\pi)\leq R^+_\epsilon,\\
\mathfrak{Max}(\pi) & \text{, if the shortest prefix }\pi^\prime\text{ of }\pi\text{ with } wgt(\pi^\prime)\not\in [R^-_\epsilon,R^+_\epsilon] \\
& \text{satisfies }wgt(\pi^\prime)>R^+_\epsilon, \\
\mathfrak{Min}(\pi) & \text{, otherwise,}
\end{cases}\]
satisfies $\mathit{PE}^\mathfrak{T}_{s_{\mathit{init}}}\geq \mathit{PE}^{\sup}_{s_{\mathit{init}}} - \epsilon$. 
\end{mythm}

\begin{proof}
Let $\mathfrak{S}$ be a weight-based deterministic scheduler with $\mathit{PE}^\mathfrak{S}_{s_{\mathit{init}}}=\mathit{PE}^{\sup}_{s_{\mathit{init}}}$. Define \[\mathfrak{T} (\pi) = 
\begin{cases}
\mathfrak{S}(\pi) & \text{, if any prefix }\pi^\prime\text{ of }\pi\text{ satisfies }R^-_\epsilon\leq wgt(\pi)\leq R^+_\epsilon,\\
\mathfrak{Max}(\pi) & \text{, if the shortest prefix }\pi^\prime\text{ of }\pi\text{ with } wgt(\pi^\prime)\not\in [R^-_\epsilon,R^+_\epsilon] \\
& \text{satisfies }wgt(\pi^\prime)>R^+_\epsilon, \\
\mathfrak{Min}(\pi) & \text{, otherwise.}
\end{cases}\]
We give an estimation for the difference $\mathit{PE}^{\sup}_{s_{\mathit{init}}}-\mathit{PE}^\mathfrak{T}_{s_{\mathit{init}}}$. In order to do so, we define the following two sets:
\begin{eqnarray*}
\Pi^+_\epsilon &:=& \{\pi \text{ finite }\mathfrak{S}\text{-path } | wgt(\pi)\geq R^+_\epsilon \\
&& \text{ and for any proper prefix }\pi^\prime \text{ of }\pi\text{, } R^-_\epsilon\leq wgt(\pi^\prime)\leq R^+_\epsilon\}, \\
\Pi^-_\epsilon &:=& \{\pi \text{ finite }\mathfrak{S}\text{-path } | wgt(\pi)\leq R^-_\epsilon \\
&& \text{ and for any proper prefix }\pi^\prime \text{ of }\pi\text{, } R^-_\epsilon\leq wgt(\pi^\prime)\leq R^+_\epsilon\}.
\end{eqnarray*}

The schedulers $\mathfrak{S}$ and $\mathfrak{T}$ agree on all paths not in $\Pi^+_\epsilon$ or $\Pi^-_\epsilon$. Hence,
\begin{eqnarray*}
&&\mathit{PE}^{\sup}_{s_{\mathit{init}}} - \mathit{PE}^\mathfrak{T}_{s_{\mathit{init}}} \\
&=& \sum_{\pi\in \Pi^+_\epsilon} \Pr^\mathfrak{S}_{s_{\mathit{init}}}(\pi)\cdot(\mathit{PE}^{\sup}_{last(\pi)}  [wgt(\pi)] - \mathit{PE}^{\mathfrak{Max}}_{last(\pi)} - p^{\max}_{last(\pi)}\cdot wgt(\pi) ) + \\
&&  \sum_{\pi\in \Pi^-_\epsilon} \Pr^\mathfrak{S}_{s_{\mathit{init}}}(\pi)\cdot(\mathit{PE}^{\sup}_{last(\pi)} [wgt(\pi)] - \mathit{PE}^{\mathfrak{Min}}_{last(\pi)} - p^{\min}_{last(\pi)}\cdot wgt(\pi)).
\end{eqnarray*}
For the first sum, we have the following estimation:
\begin{eqnarray*}
&&\sum_{\pi\in \Pi^+_\epsilon}\Pr^\mathfrak{S}_{s_{\mathit{init}}}(\pi)\cdot( \mathit{PE}^{\sup}_{last(\pi)} [wgt(\pi)] - \mathit{PE}^{\mathfrak{Max}}_{last(\pi)} - p^{\max}_{last(\pi)}\cdot wgt(\pi)) \\
&\leq&\sum_{\pi\in \Pi^+_\epsilon} \Pr^\mathfrak{S}_{s_{\mathit{init}}}(\pi)\cdot (\mathit{PE}^{\sup}_{last(\pi)} - \mathit{PE}^{\mathfrak{Max}}_{last(\pi)} ) \\
&\leq&  \Pr^\mathfrak{S}_{s_{\mathit{init}}}(\Pi^+_\epsilon) \cdot D \leq \Pr^\mathfrak{S}_{s_{\mathit{init}}} (\Diamond wgt\geq R^+_\epsilon) \cdot D \\
&\leq&  \lambda_\mathcal{M} ^{ \frac{\log(2D)+\log(1/\epsilon)}{\log(1/\lambda_\mathcal{M})}} \cdot D = 2^{\log(\lambda_\mathcal{M})\cdot  \frac{\log(2D)+\log(1/\epsilon)}{\log(1/\lambda_\mathcal{M})}} \cdot D = 2^{\log(\epsilon)-\log(2D)} \cdot D= \epsilon /2.
\end{eqnarray*}

For the second sum, consider the following scheduler.
On extensions of paths in $\Pi^-_\epsilon$, let $\mathfrak{S}^\prime$ be the scheduler which behaves like $\mathfrak{S}$ until the accumulated weight is at least $\mathfrak{q}$ again and then switches to the choices of $\mathfrak{Min}$. We know that $\mathfrak{S}$ only chooses actions in $Act^{\min}(s)$ when in a state $s$ with accumulated weight below $\mathfrak{q}$.
On the other hand, $\mathfrak{Min}$ is optimal among these schedulers. So, $\mathfrak{Min}$ is at least as good as $\mathfrak{S}^\prime$ on extensions of paths in $\Pi^-_\epsilon$ with respect to maximizing the partial expectation. Further, starting at a path in $\Pi^{-}_\epsilon$ we reach an accumulated weight of at least $\mathfrak{q}$ only if we accumulate a weight of at least $R^+$. Afterwards, we can bound the advantage of $\mathfrak{S}$ over $\mathfrak{Min}$ by $D$. So, we get the following  estimation:
\begin{eqnarray*}
&&\sum_{\pi\in \Pi-_\epsilon}\Pr^\mathfrak{S}_{s_{\mathit{init}}}(\pi)\cdot( \mathit{PE}^{\sup}_{last(\pi)} [wgt(\pi)] - \mathit{PE}^{\mathfrak{Min}}_{last(\pi)} - p^{\min}_{last(\pi)}\cdot wgt(\pi)) \\
&\leq&\sum_{\pi\in \Pi^-_\epsilon} \Pr^\mathfrak{S}_{s_{\mathit{init}}}(\pi)\cdot (\Pr^{\max}_{last(\pi)}(\Diamond wgt \geq R^+) \cdot D ) \leq   \epsilon /2. 
\end{eqnarray*}
So, $\mathit{PE}^{\sup}_{s_{\mathit{init}}}-\mathit{PE}^\mathfrak{T}_{s_{\mathit{init}}}\leq \epsilon$.
\end{proof}

This result now allows us to compute an $\epsilon$-approximation and an $\epsilon$-optimal scheduler with finite memory by linear programming, similar to the case of non-negative weights. The linear program has $R^+_\epsilon + R^-_\epsilon$ many variables and $|Act|$-times as many inequalities. 

\begin{mycor}
The maximal partial expectation $\mathit{PE}^{\sup}_{s_{\mathit{init}}}$ can be approximated up to an absolute error of $\epsilon$ in time exponential in the size of $\mathcal{M}$ and polynomial in $\log(1/\epsilon)$. 
\end{mycor}

\begin{proof}
Consider the following linear program with one variable $x_{s,w}$ for each $s\in S$ and $R^- - W\leq w \leq R^+ +W$:

Minimize $\sum_{s,w} x_{s,w}$ under the following constraints:
\[x_{\mathit{goal}, w} = w \text{, and } x_{fail, w}=0,\]
for $w\geq R^+$ and $s\in S\setminus \{\mathit{goal}, fail\}$,
\[ x_{s,w} = \mathit{PE}^\mathfrak{Max}_s +p^{\max}_s \cdot w, \]
for $w\leq R^-$ and $s\in S\setminus \{\mathit{goal}, fail\}$,
\[ x_{s,w} = \mathit{PE}^\mathfrak{Min}_s +p^{\min}_s \cdot w, \]
and for $R^-<w < R^+$, $s\in S\setminus \{\mathit{goal}, fail\}$, and $\alpha \in Act(s)$,
\[ x_{s,w}\geq \sum_{t\in S} P(s,\alpha,t)\cdot x_{t,w+wgt(s,\alpha)} .\]
The unique solvability can be shown as in Proposition \ref{prop:linearprogram} using  that all end components have negative mean payoff:
We can interpret the linear program on an MDP with state space $S\times \{R^- -W, \dots, R^+\}$ and the transitions induced by $\mathcal{M}$. This MDP now has no end components.
\end{proof}


\subsection{Transfer to Conditional Expectations} \label{app:approximationCE}


We restate the algorithm given in Section \ref{sec:approximation}.
Let $\mathcal{M}$ be an MDP with $\mathit{CE}^{sup}_{s_{\mathit{init}}}<\infty$ and let $\epsilon>0$. By Proposition \ref{prop:reduction}, we can assume that $\Pr^{\min}_{\mathcal{M},s_{\mathit{init}}}(\Diamond \mathit{goal})=:p$ is positive.
 We know that $\mathit{CE}^{sup}_{s_{\mathit{init}}}\in [\mathit{CE}^{\mathfrak{Max}}_{s_{\mathit{init}}},\mathit{CE}^{ub}]$. We perform a binary search to approximate $\mathit{CE}^{sup}_{s_{\mathit{init}}}$:

We put $A_0:=\mathit{CE}^{\mathfrak{Max}}_{s_{\mathit{init}}}$ and $B_0:=\mathit{CE}^{ub}$. Given $A_i$ and $B_i$, let $\theta_i := (A_i+B_i)/2$. Then, we approximate $\mathit{PE}^{\sup}_{s_{\mathit{init}}}[-\theta_i]$ up to an absolute error of $p\cdot \epsilon$. 
 Let $E_i$ be the value of this approximation. If $E_i \in [-2p\cdot \epsilon, 2p\cdot \epsilon]$, terminate and return $\theta_i$ as the approximation for $\mathit{CE}^{\sup}_{s_{\mathit{init}}}$. If $E_i<-2p\cdot \epsilon$, put $A_{i+1}:=A_i$ and $B_{i+1}:=\theta_i$, and repeat. If $E_i>2p\cdot\epsilon$, put $A_{i+1}:=\theta_i$ and $B_{i+1}:= B_i$, and repeat.

\begin{myprop}
The procedure terminates after at most $\lceil \log((A_0-B_0)/(p\cdot\epsilon))\rceil$ iterations and returns an $3\epsilon$-approximation of $\mathit{CE}^{\sup}_{s_{\mathit{init}}}$ in time  exponential in the size of $\mathcal{M}$ and polynomial in $\log(1/\epsilon)$.
\end{myprop}

\begin{proof}
We begin by showing that the algorithm terminates after at most $\lceil \log((A_0-B_0)/(\epsilon\cdot p)) \rceil$ many iterations, i.e. when $|A_i-B_i|\leq \epsilon\cdot p$.
We know that $\mathit{PE}^{\sup}_{s_{\mathit{init}}}[-\theta_i]<0$ if $E_i<- 2p\epsilon$ and $\mathit{PE}^{\sup}_{s_{\mathit{init}}}[-\theta_i]>0$ if $E_i>2p\epsilon$.
By Proposition \ref{prop:threshold}, we conclude that $\mathit{CE}^{\sup}_{s_{\mathit{init}}}\in [A_{i+1},B_{i+1}]$ at any time.
So, after at most $log((A_0-B_0)/(\epsilon\cdot p))$ many iteration, we have that $|A_i-B_i|\leq \epsilon\cdot p$ and hence $\mathit{CE}^{\sup}_{s_{\mathit{init}}}-\epsilon\cdot p \leq \theta_i \leq \mathit{CE}^{\sup}_{s_{\mathit{init}}}+\epsilon\cdot p$. We claim that then $E_i \in [-2p\epsilon, 2p\epsilon]$. Suppose $E_i < -2p\epsilon$. Then $\mathit{PE}^{\sup}_{s_{\mathit{init}}}[-\theta_i] < - p\epsilon$. But we have \[0=\mathit{PE}^{\sup}_{s_{\mathit{init}}}[-\mathit{CE}^{\sup}_{s_{\mathit{init}}}] \leq \mathit{PE}^{\sup}_{s_{\mathit{init}}}[-\theta_i+p\epsilon] \leq \mathit{PE}^{\sup}_{s_{\mathit{init}}}[-\theta_i] + p\epsilon\] contradicting the supposition. Analogously, we show that $E_i$ cannot be greater than $2p\epsilon$.

Next, we show that the algorithm returns an $3\epsilon$-approximation of $\mathit{CE}^{\sup}_{s_{\mathit{init}}}$. As soon as the algorithm terminates, we have that $E_i \in [-2p\epsilon, 2p\epsilon]$. So, $\mathit{PE}^{\sup}_{s_{\mathit{init}}}[-\theta_i]  \in [-3p\epsilon, 3p\epsilon]$. So there is a scheduler $\mathfrak{S}$ with  
\[\mathit{PE}^{\mathfrak{S}}_{s_{\mathit{init}}}[-\theta_i] =  \mathit{PE}^{\mathfrak{S}}_{s_{\mathit{init}}}-\theta_i \cdot \Pr^{\mathfrak{S}}_{s_{\mathit{init}}}(\Diamond \mathit{goal})  \geq -3p\epsilon.\] 
As $\Pr^{\mathfrak{S}}_{s_{\mathit{init}}}(\Diamond \mathit{goal}) \geq p$, this implies $\mathit{CE}^{\sup}_{s_{\mathit{init}}} \geq \mathit{CE}^{\mathfrak{S}}_{s_{\mathit{init}}}\geq \theta_i -3\epsilon$.
On the other hand, suppose that $\mathit{CE}^{\sup}_{s_{\mathit{init}}}>\theta_i+3\epsilon$. Then there is a scheduler $\mathfrak{T}$ with $\mathit{CE}^{\mathfrak{T}}_{s_{\mathit{init}}}>\theta_i +3\epsilon$. For this scheduler, we have
\[0 < \mathit{PE}^{\mathfrak{T}}_{s_{\mathit{init}}}[-\theta_i-3\epsilon] = \mathit{PE}^{\mathfrak{T}}_{s_{\mathit{init}}}[-\theta_i] - 3\epsilon \cdot \Pr^{\mathfrak{T}}_{s_{\mathit{init}}}(\Diamond \mathit{goal}) \leq \mathit{PE}^{\mathfrak{T}}_{s_{\mathit{init}}}[-\theta_i] - 3\epsilon\cdot p .\]
This contradicts $\mathit{PE}^{\sup}_{s_{\mathit{init}}}[-\theta_i] \leq 3p\epsilon$.
Therefore, the algorithm indeed returns a $3\epsilon$-approximation of $\mathit{CE}^{\sup}_{s_{\mathit{init}}}$.

Finally, we show that the claimed running time is correct: The algorithm stops after at most $\lceil \log((A_0-B_0)/(\epsilon\cdot p)) \rceil$ iterations. As all values involved can be computed in polynomial time, this is polynomial in the size of $\mathcal{M}$ and linear in $\log(1/\epsilon)$. In each iteration, we have to approximate the maximal partial expectation $\mathit{PE}^{\sup}_{s_{\mathit{init}}}[-\theta_i]$ up to an absolute error of $p\cdot \epsilon$. As the logarithmic length of $\theta_i$ is polynomial in the size of $\mathcal{M}$ as well, this can be done in time exponential in the size of $\mathcal{M}$ and polynomial in $\log(1/\epsilon)$.
\end{proof}


\end{document}